\documentclass[3p,times,procedia]{elsarticle}
\usepackage{ecrc}
\volume{00}
\firstpage{1}
\runauth{Qi \& Ta}
\jid{aej}

\usepackage{amsmath,amssymb,, amsthm, amscd,}
\usepackage{graphicx}
\usepackage{algorithm}
\usepackage{algpseudocode}
\usepackage{booktabs}
\usepackage{float}
\usepackage{caption}
\usepackage{subcaption}
\usepackage{hyperref}
\usepackage{siunitx}
\usepackage{multirow}
\usepackage{url}
\usepackage{xcolor}

\newtheorem{theorem}{Theorem}[section]

\makeatletter
\def\ps@pprintTitle{\let\@oddhead\@empty\let\@evenhead\@empty
\def\@oddfoot{\let\@journalShortName\@empty}\let\@evenfoot\@oddfoot}
\makeatother

\begin{document}
\begin{frontmatter}
\title{Modeling Predator-Prey Dynamics with Stochastic Differential Equations: Patterns of Collective Hunting and Nonlinear Predation Effects}
\author[]{Junyi Qi}
\author[]{T\^{o}n Vi\d {\^{e}}t T\d{a}}
\address{Graduate School of Bioresource and Bioenvironmental Sciences, Kyushu University\\
744 Motooka, Nishi Ward, Fukuoka 819-0395, Japan}

\begin{abstract}
We investigate predator–prey school interactions in aquatic environments using a stochastic differential equation (SDE)-based, particle-level model that incorporates attraction, repulsion, alignment, and environmental noise. Two predation strategies—center attack and nearest attack—are examined to assess their effects on prey survival, predator efficiency, and group dynamics. Simulations reveal diverse emergent behaviors such as prey dispersal and regrouping, oscillatory predation with collective defense, and predator encirclement. Results show that collective hunting enhances capture efficiency compared to solitary attacks, but benefits diminish beyond a critical predator group size due to intra-predator competition. This work provides new insights into cooperative predation and introduces a generalizable SDE framework for analyzing predator–prey interactions.
\end{abstract}

\end{frontmatter}

\section{Introduction} \label{introduction}
Predator-prey interactions are fundamental drivers of ecosystem stability, population dynamics, and evolutionary adaptation \cite{Allesina2012}. In aquatic environments, these interactions are further complicated by the prevalence of schooling behavior in prey species \cite{Partridge1982}. While schooling offers numerous advantages to prey, including enhanced hydrodynamic efficiency \cite{Jolles2020, Ioannou2021, Rubenstein1978}, increased mating opportunities \cite{Jolles2020, Rubenstein1978}, improved foraging success \cite{Jolles2020, Ioannou2021, Pitcher1982}, and, critically, reduced predation risk \cite{Partridge1982, Jolles2020, Rubenstein1978, Pitcher1986, Magurran1990, Demsar2015}, predators have, in turn, evolved cooperative hunting strategies to counteract these defenses and increase their hunting success \cite{Anderson2001, Dugatkin1997, Gazda2004}.

Because of these dynamics, swarm behavior has attracted considerable attention across ecology, mathematical modeling, and computational science \cite{Aoki1982, Reynolds1987, Cucker2007, Couzin2009, Vicsek2012, Ton2018}. Prior studies have sought to  elucidate how simple inter-individual interactions give rise to complex collective behaviors, a phenomenon of particular relevance in predator-prey systems where both collective tactics and cooperation are crucial for survival and predation success.  Mathematical models provide a quantitative framework for exploring these behaviors, enabling researchers to reproduce and predict emergent group-level patterns across varying environments.  The capacity of such models to characterize the intricate, underlying interactions within groups is essential for understanding aquatic ecosystems, where individual interactions and coordination underpin both group efficiency and resilience.

Several foundational mathematical models have significantly advanced our understanding of swarm dynamics and predator-prey interactions.  The Vicsek model, introduced by Vicsek et al. \cite{Vicsek1995}, is renowned for its ability to simulate self-organized motion in swarms.  In this model, agents (or particles) move at a constant speed but adjust their direction based on the average heading of their neighbors.  This simple rule set demonstrates how local interactions can lead to large-scale collective alignment in the presence of noise.  The underlying concept—that prey species (e.g., schooling fish) develop cohesion strategies under predation pressure—provided an early framework for studying collective defense mechanisms in ecological systems.

The Boid model, developed by Reynolds \cite{Reynolds1987} in 1986, represents another significant contribution to the study of swarm dynamics.
This model describes how entities (e.g., birds, fish) move based on three fundamental rules:
separation (avoiding collisions), alignment (matching movement direction), and cohesion (staying near neighbors). These simple, local rules give rise to complex, coordinated group movements, demonstrating key mechanisms of collective behavior and swarm dynamics.

The model developed by Oboshi et al. \cite{Oboshi2002} employs genetic algorithms to evolve prey escape behaviors in a simulated ecological environment.  Their model incorporates parameters representing individual behaviors—attraction ($\alpha$), combined strategies ($\beta$), alignment ($\gamma$), and avoidance ($\delta$)—which evolve through selection, crossover, and mutation.  Upon predator detection, prey may switch to an "urgent mode," adopting adaptive movement strategies such as dispersal or cohesive schooling to enhance evasion.  These strategies leverage the dilution and confusion effects to improve survival.  Through simulations across various scenarios, Oboshi et al. demonstrate that individuals exhibiting collective schooling behaviors achieve significantly higher evasion success rates compared to solitary individuals.  Their work also explores the evolution of cohesion through a selective mechanism, where cohesion confers a survival advantage under predation pressure.

Two primary classes of mathematical models are commonly employed to study predator-prey interactions that incorporate schooling behavior.  The first class focuses on tracking the population densities of predators and prey over time.  Within this framework, interactions and dynamics are typically described by models based on functional responses (e.g., \cite{Abrams2000, Abrams20002, Bera2015, Bera2016, Maiti2016, Manna2018, Manna2019,Yasin2025,Yasin2024,Baber2024,Ahmed2025,ta2010survival}).  While effective for understanding macro-level trends, this population-level approach lacks the granularity necessary to capture individual-based behaviors and the dynamics of complex schooling.

The second class of models, which forms the basis for the present study, examines predator-prey systems from the perspective of self-propelled agents within an artificial ecosystem.  Here, both predators and prey are treated as independent agents (or particles) that move and interact within a defined spatial domain.  Agent-based modeling enables detailed investigation of individual-level interactions, crucial for understanding the emergence of schooling as an adaptive strategy under predation pressure.  By modeling agents as self-driven entities governed by simple interaction rules—including alignment, attraction, repulsion, and evasion—we can simulate emergent group-level phenomena such as collective motion, predator avoidance, and schooling behavior.

To justify the modeling framework presented in this paper, it is important to acknowledge prior work documenting group hunting in marine predators.  These studies provide crucial context regarding predator cooperation in overcoming prey defenses and enhancing hunting success.  Real-world observations serve as a foundation for validating models that simulate predator-prey interactions and schooling behaviors.  Observed cooperative strategies in marine predators underscore the influence of coordination and collective action, further emphasizing the need for an agent-based modeling approach to accurately capture these dynamics.  Several key studies, referenced below, offer empirical evidence supporting the existence of cooperative hunting behaviors, thus aligning with and reinforcing the theoretical model employed in this research.

Herbert-Read et al. \cite{Herbert-Read2016} provide compelling observations of group hunting in sailfish, revealing a cooperative strategy termed "proto-cooperation."  Sailfish coordinate attacks on sardine schools through an alternating role strategy, resulting in injury to multiple prey during each attack.  Although only a fraction of attacks are successful, increased prey injury correlates with higher post-attack capture rates.  This alternating attack strategy, which does not require spatial coordination between individual hunters, enhances hunting efficiency, benefiting individual predators and demonstrating the advantages of cooperative hunting through shared effort.

Steinegger et al. \cite{Steinegger2018} investigated cooperative foraging in yellow saddle goatfish, demonstrating that the success of group foraging relies on simple decision rules.  In a hunting context, individuals dynamically switch between "initiator" and "follower" roles based on their spatial position relative to the prey.  The lead pursuer directly chases the prey, while others intercept escape routes, dramatically increasing capture success. This example illustrates how simple, spatially oriented decisions can generate complex group coordination and highlights the potential of simple behavioral rules in cooperative hunting scenarios.

Lönnstedt et al. \cite{Lonnstedt2014} studied cooperative hunting in lionfish (Dendrochirus zebra), describing how they use flared fin displays to communicate with partners.  Signaling among conspecifics and heterospecifics triggers cooperation in coordinated group hunts.  Their findings demonstrate that cooperative hunting significantly increases success rates compared to solitary hunting.  Furthermore, the observed initiator-responder, turn-taking attack strategy aligns more closely with balanced resource allocation than individual utility maximization.  This suggests that cooperative signaling is an inherent trait in lionfish, contributing to their predatory efficiency.

Strübin et al. \cite{Strubin2011} described collaborative hunting in yellow saddle goatfish (Parupeneus cyclostomus), highlighting role differentiation within group hunts in coral habitats.  Some individuals accelerate to intercept mobile prey, while others strategically position themselves to block escape routes. This coordinated behavior, involving prey encirclement and the use of barbels to dislodge prey from coral crevices, likely reflects the challenges of capturing mobile prey in complex environments.  Such hunting strategies may also contribute to the evolution of group living in this species.

Major \cite{Major1978} investigated predator-prey interactions between jack (Caranx ignobilis) and Hawaiian anchovy (Stolephorus purpureus).  His findings demonstrate that clustered predators are more effective than solitary predators at capturing prey, especially when prey form schools.  Typically, the leading predator exhibits the highest capture success, while following predators provide secondary benefits by further disrupting prey schools and isolating individuals for easier capture.  These findings support the concept of co-evolution between predator and prey, where schooling behavior has shaped both predation efficiency and prey survival through natural selection.

Demšar and Lebar Bajec \cite{Demsar2014} employed a fuzzy logic model to investigate hunting strategies, including attacking the school's center or targeting the nearest individuals. They observed that predators face challenges due to the coordinated behavior of prey.  Subsequently, Demšar et al. \cite{Demsar2015} highlighted the potential for analyzing composite hunting strategies, enabling predators to adapt their tactics over time.  For example, predators might initially disperse prey schools before launching more efficient attacks. This work underscores the importance of coordinated, flexible actions within prey schools in countering predator defenses and influencing hunting efficiency.

Nevertheless, most existing models remain limited. Many predator–prey systems either focus on population-level densities using functional responses \cite{Abrams2000, Bera2016}, or on interactions between a single predator and a prey school \cite{Aditya2024}. Such approaches, while informative, cannot fully capture the complexity of group hunting observed in natural systems, where predators coordinate their movements in schools to increase capture success \cite{Herbert-Read2016, Lonnstedt2014}. In addition, several agent-based frameworks omit biologically realistic constraints such as collision avoidance, which can produce unrealistic dynamics. These gaps motivate the present work.

In parallel with research on ecological modeling, recent years have witnessed a rapid expansion of stochastic and fractional dynamical models across diverse scientific fields. Such approaches have been applied to epidemic dynamics, for instance in stochastic COVID-19 models with Lévy noise and isolation strategies \cite{Danane2021,Hama2022}, as well as in the broader context of pandemic modeling and intelligent control \cite{Hammouch2023}. Similar methods have been employed in tumor–immune interactions \cite{Baleanu2019}, water pollution management \cite{Ebrahimzadeh2024}, nonlinear oscillatory systems such as time-dependent pendula \cite{Baleanu2023}, and the transmission of emerging infectious diseases  \cite{gao2025dynamics,Baleanu2023b}. Collectively, these studies highlight a growing interest in the use of stochasticity, jumps, and fractional operators to enhance model realism and predictive capacity, particularly in systems where uncertainty, noise, and memory effects are intrinsic. In this broader context, our work contributes to the understanding of predator–prey interactions in aquatic environments, where variability and stochastic influences also play critical roles in shaping group-level dynamics.

Accordingly, this study develops a stochastic differential equation (SDE) model that explicitly describes the dynamics of predator and prey schools. While the framework simplifies certain processes such as sensory perception and decision-making, its core assumptions---alignment, attraction, repulsion, and behavioral uncertainty---are grounded in established principles of fish school behavior \cite{Camazine2001,Couzin2009,Krause2002}. 
An additional key assumption concerns the hunting strategies of predators. Drawing on observations in various fish species \cite{Pavlov2000,Grobecker1983,Rizkalla2008,Pietsch2020,Shallenberger1973}, we introduce the following rule: a predator determines its movement based on a weighted average of the position and orientation of either (i) its nearest prey fish or (ii) the centroid of the prey school. The first strategy, termed \emph{nearest attack}, reflects a predator focusing primarily on its closest prey, whereas the second strategy, termed \emph{center attack}, represents predators targeting the geometric center of the prey school. 
Finally, to enhance biological realism, the model incorporates interaction limits that prevent overlap between individuals. Together, these features ensure a balance between ecological plausibility and mathematical tractability.

The proposed model generates diverse predation and escape patterns consistent with observations in natural systems, such as dolphins or tuna coordinating against fish schools \cite{Krause2002}. Moreover, it demonstrates the advantages of cooperative predation compared with solitary hunting. In doing so, it provides a new agent-based modeling framework that bridges individual-level interactions and school-level dynamics. This contribution helps the scientific community in three main ways: (i) by advancing theoretical understanding of how predator strategies and group size affect prey survival and predator efficiency; (ii) by offering a flexible platform to test ecological and evolutionary hypotheses about collective hunting and defense; and (iii) by supporting applications in conservation biology and bio-inspired computation. Beyond reproducing known behaviors, the framework thus provides both explanatory power and practical utility for researchers studying collective dynamics in ecological and artificial systems.

The remainder of the paper is structured as follows. Section \ref{modeldescription} introduces the SDE-based model and details the hunting strategies and capture conditions. Section \ref{PredationPatternsBetweenPredatorandPreySchools} examines the influence of model parameters on emergent predation patterns. Section \ref{Effectsofpredatorschoolsize} investigates the effect of predator school size on prey dynamics. Finally, Section \ref{Conclusion} summarizes the main findings and discusses their implications.

\section{Model Description} \label{modeldescription}

This section presents our model for predator–prey fish schooling. The model captures both \textit{intra-school interactions} (within predator and prey schools) and \textit{inter-school interactions} (between predators and prey, either at the individual level or relative to the prey school centroid). Together, these interactions define the behavioral rules that govern schooling and determine the forces exerted by predators on individual prey or on the prey school as a whole.

Predator–prey interactions are based on behavioral rules observed by biologists in natural schooling dynamics:

\begin{center}
\begin{enumerate}[(a)]
    \item Prey schools lack a designated leader; all individuals follow the same rules.
    \item Each prey fish adjusts its movement based on a weighted average of the positions and orientations of its nearest neighbors.
    \item Each predator follows a similar rule, also without a leader, using a weighted average of the positions and orientations of its nearest neighbors.
    \item Each prey fish selects its escape direction relative to the weighted positions of nearby predators.
    \item Both predator and prey behaviors include a stochastic component, reflecting randomness in hunting.
    \item Each predator targets either the nearest prey or the prey school centroid, depending on its hunting strategy.
\end{enumerate}
\end{center}

Rules (a), (b), (c), and (e) follow \cite{Camazine2001,Takeshi2012,b7}, while (d) is introduced in \cite{Aditya2024}. Rule (f) corresponds to the “nearest-prey” and “centroid” attack strategies, mentioned in Section \ref{introduction}.

\subsection*{Stochastic Differential Equation (SDE) System}

The system of SDEs governing the predator–prey dynamics is given by:
\begin{equation}
    \left\{
    \begin{aligned}
        dx_{i} &= v_{i} \, dt + \sigma_{i} \, dw_{i}^1(t), \quad i = 1, 2, \ldots, N,\\
        dv_{i} &= \Bigg[ -\alpha\sum_{\substack{j=1 \\ j\neq i}}^{N}\left( \frac{r^p}{\|x_i-x_j\|^p+ \epsilon } - \frac{r^q}{\|x_i-x_j\|^q+ \epsilon } \right)(x_{i}-x_{j}) \\
               &\quad -\beta\sum_{\substack{j=1 \\ j\neq i}}^{N}\left( \frac{r^p}{\|x_i-x_j\|^p+ \epsilon } + \frac{r^q}{\|x_i-x_j\|^q+ \epsilon } \right)(v_{i}-v_{j}) \\
               &\quad + \delta\sum_{k=1}^{M}\left( \frac{R_{1}^{\theta_{1}}}{\|x_i-y_k\|^{\theta_{1}}+ \epsilon } (x_{i}-y_{k}) \right) \Bigg] dt, \\
        dy_{k} &= u_{k} \, dt + \sigma_{k} \, dw_{k}^2(t), \quad k = 1, 2, \ldots, M,\\
        du_{k} &= \left[ \delta\sum_{\substack{j=1 \\ j\neq k}}^{M} \frac{R^\theta}{\|y_{k}-y_{j}\|^\theta+\epsilon } (y_{k}-y_{j}) + F(\mathbf{X}, \mathbf{V}, \mathbf{Y}, \mathbf{U}) \right] dt,
    \end{aligned}
    \right.
\label{eq:dynamic_predation_system}    
\end{equation}

where
\begin{itemize}
    \item \( x_i, v_i \): position and velocity of the \( i \)-th prey,
    \item \( y_k, u_k \): position and velocity of the \( k \)-th predator, in $\mathbb R^d$ ($d=2,3,\dots$),
    \item \( N \): prey population size, \quad \( M \): predator population size,
    \item \( \mathbf{X}= (x_1, x_2, \dots, x_N), \mathbf{V}=(v_1, v_2, \dots, v_N), \mathbf{Y}=(y_1, y_2, \dots, y_M), \mathbf{U}=(u_1, u_2, \dots, u_M) \): state vectors of all prey and predator positions and velocities,
    \item \( \| \cdot \| \): Euclidean norm,
    \item $1<p<q<\infty$, \( \theta>1 \), \( \theta_1>1\): exponents, \quad other parameters: positive constants.
\end{itemize}

This system has two main components: prey dynamics and predator dynamics, coupled through predator–prey interactions.

\paragraph{Prey school dynamics}  
The prey dynamics, governed by $x_i$ and $v_i$, incorporate attraction, repulsion, velocity alignment, and stochastic noise.  
\begin{itemize}
    \item If $\|x_i-x_j\|>r$, the attractive force $-\alpha \tfrac{r^p(x_i-x_j)}{\|x_i-x_j\|^p+\epsilon}$ dominates, drawing fish together.  
    \item If $\|x_i-x_j\|<r$, the repulsive force $\alpha \tfrac{r^q(x_i-x_j)}{\|x_i-x_j\|^q+\epsilon}$ dominates, preventing collisions.  
\end{itemize}

The predator–prey repulsion is modeled by
\[
\delta\sum_{k=1}^{M}  \frac{R_1^{\theta_1}}{\|x_i - y_k\|^{\theta_1}+ \epsilon} (x_i - y_k),
\]
where \(R_1 > r\) represents the escape distance. The force increases as predators approach, mimicking natural escape behavior.  

Velocity matching is also distance-dependent: individuals align more strongly when closer, reducing collisions. These interaction forces are generalized from Newton’s law of gravitation (attraction) and van der Waals forces (repulsion).

\paragraph{Predator school dynamics}  
Predator dynamics, governed by $y_k$ and $u_k$, include intra-school repulsion and inter-school hunting interactions. Intra-school repulsion is modeled as:
\[
\delta \sum_{\substack{j=1 , j \neq k}}^{M} \frac{R^{\theta}}{\|y_k - y_j\|^{\theta}+ \epsilon} (y_k - y_j),
\]
with $R>0$ denoting the predator–predator critical distance. This prevents clustering and maintains spacing within the predator school.

\paragraph{Regularization via $\epsilon$}  
We incorporate a small parameter $\epsilon$ in the denominators to prevent singularities in the forces and stabilize simulations. Biologically, this modification ensures that forces remain finite at very small distances, reflecting realistic interaction limits. With $\epsilon$, the forces become
\begin{equation} \label{two_forces}
    f_1(x) = \frac{x}{x^p + \epsilon}, \quad f_2(x) = \frac{x}{x^q + \epsilon},
\end{equation}
where $x \geq 0$ is the distance between two individuals. As shown in Figure \ref{force}, these forces decay to zero at large distances, remain finite at very small distances, and vanish below a threshold, consistent with biological interactions.

\begin{figure}[H]
\begin{center}
\includegraphics[scale=0.5]{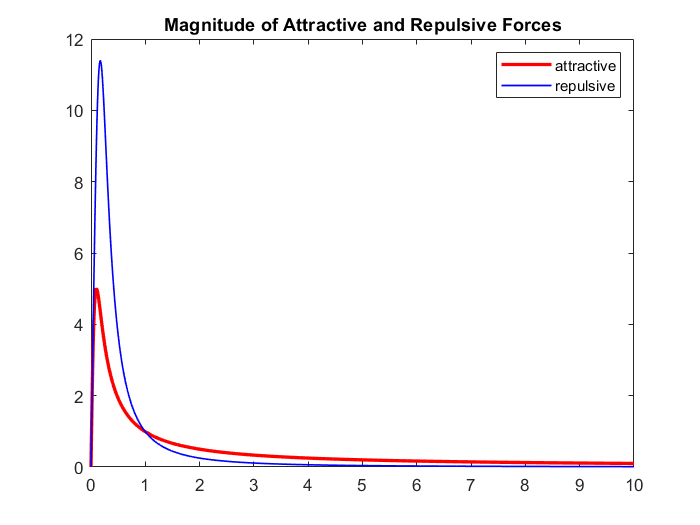}
\caption{Graphs of $f_1(x)$ and $f_2(x)$ in \eqref{two_forces}, with parameters  $p=2$, $q=3$, and $\epsilon=0.01.$ }
\label{force}
\end{center}
\end{figure}

\paragraph{Stochastic effects}  
The terms $dw_i^1(t)$ and $dw_k^2(t)$ represent independent $d$-dimensional Brownian motions, scaled by noise coefficients $\sigma_i$ and $\sigma_k$. These capture randomness in behavior due to imperfect information processing and execution. Formally, they are defined on a complete probability space $(\Omega, \mathcal F, \{\mathcal F_t\}_{t\geq 0}, \mathbb P)$ satisfying the usual conditions.

\paragraph{Hunting strategies}  
Predators in our model employ two distinct hunting strategies derived from observed behaviors in aquatic ecosystems: (i) \textit{center attack}, where predators target the centroid of the prey school, and (ii) \textit{nearest attack}, where predators pursue the closest individual prey. These strategies reflect contrasting approaches to predation, either disrupting the collective structure of the school or isolating vulnerable individuals. In the model \eqref{eq:dynamic_predation_system}, 
the inter-school interaction $F(\mathbf{X}, \mathbf{V}, \mathbf{Y}, \mathbf{U})$ represents hunting strategies.

\medskip
\textbf{Strategy I: Center attack.}  
In many aquatic ecosystems, predators target the center of a prey school to maximize capture efficiency. This approach disrupts the cohesion of the prey school, making it easier to isolate individuals \cite{Demsar2015}. Such behavior has been documented in species including mackerel (\textit{Scomber} spp.), sardines (\textit{Sardinops} spp.), and anchovies (\textit{Engraulis} spp.) \cite{Pavlov2000}.

The corresponding force function is defined as:
\begin{equation}
F(\mathbf{X}, \mathbf{V}, \mathbf{Y}, \mathbf{U})=-\frac{R_2^{\theta_2}}{\|y_k - x_c\|^{\theta_2}+ \epsilon } \left[\gamma_1(y_k - x_c)+\gamma_1\gamma_2(u_k-v_c)\right],
\label{eq:center}
\end{equation}
where the prey school center is given by the mean position and velocity:
\[
x_c = \frac{1}{N} \sum_{i=1}^{N} x_i, 
\quad v_c = \frac{1}{N} \sum_{i=1}^{N} v_i.
\]
Here, $R_2$ is the threshold distance beyond which predators initiate attacks, with $R_2>r$. Parameters $\theta_2$, $\gamma_1$, and $\gamma_2$ are positive constants. The distance $\|y_k - x_c\|$ determines the strength of attraction: predators closer to the school center experience stronger pulling forces toward it.

\medskip
\textbf{Strategy II: Nearest attack.}  
In contrast, some predators prioritize the closest prey, reflecting ambush predation tactics observed in marine species such as stonefish (\textit{Synanceia verrucosa}) \cite{Grobecker1983}, Atlantic stargazer fish (\textit{Uranoscopus scaber}) \cite{Rizkalla2008}, frogfish (\textit{Antennarius commerson}) \cite{Pietsch2020}, and scorpionfish (\textit{Iracundus signifer}) \cite{Shallenberger1973}.

The force function for this strategy is:
\begin{equation}  
F(\mathbf{X}, \mathbf{V}, \mathbf{Y}, \mathbf{U}) = -\frac{1}{N} \sum_{j=1}^{N} \frac{R_2^{\theta_2}}{\|y_k - x_j\|^{\theta_2} + \epsilon } \left[ \gamma_1 (y_k - x_j) + \gamma_1 \gamma_2 (u_k - v_j) \right].
\label{eq:nearest}
\end{equation}
The parameters mirror those in the center attack strategy. However, the weighting factor $\tfrac{1}{\|y_k - x_j\|^{\theta_2}+\epsilon}$ ensures stronger attraction toward closer prey. As a result, the nearest prey is the most likely to be pursued and captured.

Under either hunting strategy \eqref{eq:center} or \eqref{eq:nearest}, we obtain the following well-posedness result.  

\begin{theorem} \label{theorem1}
Let the initial state be
\[
(x_1(0),\dots,x_N(0), v_1(0),\dots,v_N(0), y_1(0),\dots,y_M(0), u_1(0),\dots,u_M(0)) \in \mathbb{R}^{Nd} \times \mathbb{R}^{Nd} \times \mathbb{R}^{Md} \times \mathbb{R}^{Md}.
\]
Then the system \eqref{eq:dynamic_predation_system} admits a unique local solution on some interval $[0,\tau)$, where $\tau \leq \infty$. If $\tau < \infty$ almost surely, then $\tau$ is an explosion time.
\end{theorem}

\begin{proof}
The right-hand side of \eqref{eq:dynamic_predation_system} consists of functions that are locally Lipschitz continuous in 
\[
\mathbb{R}^{Nd} \times \mathbb{R}^{Nd} \times \mathbb{R}^{Md} \times \mathbb{R}^{Md}.
\]
By standard results for stochastic  differential equations, this ensures the existence and uniqueness of a local solution up to the maximal interval of existence $[0,\tau)$; see, for example, \citep{Friedman2006}.
\end{proof}

\paragraph{Initial conditions}  
The system \eqref{eq:dynamic_predation_system} is initialized with
\[
(x_1(0),\dots,x_N(0), v_1(0),\dots,v_N(0), y_1(0),\dots,y_M(0), u_1(0),\dots,u_M(0)) \in  \mathbb R^{Nd}\times \mathbb R^{Nd}\times \mathbb R^{Md}\times \mathbb R^{Md}.
\]
To ensure biologically realistic prey schools, initial prey states are generated by first simulating the prey-only system (no predators) until the fish form a cohesive school. This stable configuration is then used as the starting condition for predator–prey simulations.

\paragraph{Eaten condition}  
A prey fish is considered “eaten” once it enters a critical capture radius around any predator. If multiple prey simultaneously meet this condition, they are all removed from the system.

Formally, let $\tau_i^{N,M}$ denote the stopping time at which the $i$-th prey is captured by one of the $M$ predators:
\[
\tau_i^{N,M} = \inf \left\{ t \geq 0 : \min_k \| y_k(t) - x_i(t) \| < d_{\text{capture}} \right\},
\]
where $d_{\text{capture}}$ is the predefined capture threshold. The earliest capture time across all prey is then
\[
\tau^{N,M} = \min_i \tau_i^{N,M}.
\]

As prey are captured, their population decreases while predator numbers remain constant. Over time, the system evolves from $N$ prey and $M$ predators to $(N-N^*)$ prey and $M$ predators, where $N^*$ is the number of prey removed.  

This process is implemented in Algorithm~\ref{alg:scheme3}, which iteratively updates positions and velocities while checking whether prey satisfy the eaten condition at each simulation step.

\begin{algorithm}
\caption{Euler scheme for simulating predator–prey interactions.}
\label{alg:scheme3}
\textbf{Input:} 
$N$ (number of prey), 
$M$ (number of predators), 
$t$ (time vector), 
$x_0, v_0$ (initial prey positions and velocities), 
$y_0, u_0$ (initial predator positions and velocities), 
$\text{attack\_mode} \in \{\text{center}, \text{nearest}\}$ \\
\textbf{Output:} 
Trajectories of prey and predators, capture log
\begin{algorithmic}[1]
\Function{EulerMethod}{$N, M, t, x_0, v_0, y_0, u_0, \text{attack\_mode}$}
    \State Initialize arrays for prey states $(x,v)$ and predator states $(y,u)$
    \State Set initial conditions $(x(0),v(0),y(0),u(0)) \gets (x_0,v_0,y_0,u_0)$
    \State Initialize capture log $\mathcal{L} \gets \emptyset$, captured set $\mathcal{C} \gets \emptyset$

    \For{each time step $n=1,\dots, |t|$}
        \State Compute prey forces $F^{\text{prey}}$
        \If{$\text{attack\_mode} = \text{center}$}
            \State Compute predator forces $F^{\text{pred}}$ using center-attack rule (Eq.~\ref{eq:center})
        \Else
            \State Compute predator forces $F^{\text{pred}}$ using nearest-attack rule (Eq.~\ref{eq:nearest})
        \EndIf
        \State Update prey states $(x,v)$ via Euler scheme with noise 
        \State Update predator states $(y,u)$ via Euler scheme with noise 

        \For{each predator $k$ and prey $j \notin \mathcal{C}$}
            \If{$\|y_k - x_j\| < d_{\text{capture}}$}
                \State Mark prey $j$ as captured: $\mathcal{C} \gets \mathcal{C} \cup \{j\}$
                \State Log capture event $(t_n, j, k)$ in $\mathcal{L}$
            \EndIf
        \EndFor
    \EndFor
    \State \Return trajectories $(x,v,y,u)$ and capture log $\mathcal{L}$
\EndFunction
\end{algorithmic}
\end{algorithm}

\section{Predation Patterns}
\label{PredationPatternsBetweenPredatorandPreySchools}
We first examine the emergent predation patterns arising from the model \eqref{eq:dynamic_predation_system}. By varying parameters under the two hunting strategies and capture rules, the simulations reveal distinct behavioral patterns that closely mirror predator–prey interactions observed in nature.

The baseline setup consists of a prey school of \( N = 40 \) individuals and a predator group of \( M = 3 \) in a three-dimensional domain \(\mathbb{R}^3\). In one case, to better illustrate the surrounding phenomenon, the number of predators is increased to \( M = 6 \). Noise levels are set to \( \sigma_k = 0.016 \) for predators and \( \sigma_i = 0.015 \) for prey. Two strategies are considered—center attack and nearest attack—over a total simulation time of \( t_{\max} = 3000 \) with time step \( \Delta t = 0.01 \).

Across these settings, seven characteristic predation and avoidance patterns emerge (Table~\ref{tab:pattern_strategies}). The detailed parameter choices for each case are provided in Table~\ref{tab:pattern_parameters_detailed}. Below, we describe each pattern in turn, highlighting both the simulation dynamics and their ecological relevance.

\begin{table}[H]
\centering
\begin{tabular}{@{}ll@{}}
\toprule
\text{Pattern}   & \text{Strategy} \\ \midrule
1. Pursuit Failure         & I, II \\
2. Dispersal by Predators         & I \\
3. Selective Pursuit After Dispersal        & II \\
4. Fragmentation and Regrouping       & I \\
5. Edge Dispersal and Selfish Escape        & II \\
6. Oscillatory Predation with Cohesive Defense       & I, II \\
7. Predator Surrounding        & I, II \\ \bottomrule
\end{tabular}
\caption{Predation strategies associated with each observed pattern.}
\label{tab:pattern_strategies}
\end{table}

\begin{table}[H]
\centering
\begin{tabular}{@{}llccccccccccccc@{}}
\toprule
\text{Pattern} & \text{Strategy} & $\alpha$ & $\beta$ & $p$ & $q$ & $\gamma_1$ & $\gamma_2$ & $\delta$ & $\theta$ & $\theta_1$ & $\theta_2$ & $R$ & $R_1$ & $R_2$ \\ \midrule
1 & I  & 20  & 0.5  & 2 & 5 & 10 & 24  & 0.05 & 1   & 1 & 1   & 2 & 2   & 2 \\
  & II & 15  & 0.5  & 2 & 5 &  8 & 20  & 0.05 & 1   & 1 & 1   & 2 & 2   & 2 \\ \midrule
2 & I  & 0.7 & 0.03 & 4 & 6 &  5 &  1  & 5    & 1.5 & 2 & 1   & 2 & 5   & 2 \\ \midrule
3 & II & 0.1 & 0.05 & 2 & 5 &  3 &  1  & 1.2  & 0.9 & 1 & 2   & 2 & 2   & 10 \\ \midrule
4 & I  & 0.7 & 0.03 & 4 & 6 &  5 &  1  & 5    & 1.5 & 2 & 1   & 2 & 2   & 2 \\ \midrule
5 & II & 0.9 & 0.03 & 4 & 6 &  5 & 1.2 & 5    & 1.5 & 2 & 1   & 2 & 2   & 2 \\ \midrule
6 & I  & 10  & 0.1  & 4 & 5 & 1.3& 0.2 & 0.1  & 2   & 1 & 0.3 & 2 & 2   & 2 \\
  & II &  1  & 0.1  & 4 & 5 & 1.5& 0.2 & 0.1  & 2   & 1 & 0.3 & 2 & 2   & 2 \\ \midrule
7 & I  & 10  & 0.1  & 4 & 5 & 1.3& 0.2 & 0.1  & 2   & 1 & 0.3 & 10& 2   & 2 \\
  & II & 10  & 0.1  & 4 & 5 & 1.3& 0.2 & 0.1  & 2   & 1 & 0.3 & 10& 2   & 2 \\ \bottomrule
\end{tabular}
\caption{Parameter settings for each predation pattern.}
\label{tab:pattern_parameters_detailed}
\end{table}

\noindent\textbf{Pattern 1: Pursuit Failure.}  

This pattern is evident when the prey school successfully maintains a safe distance from the predators, occurring with both the center attack and nearest attack strategies (see Figure~\ref{fig:pursuit_failure}).

In both the center attack strategy (first row of Figure~\ref{fig:pursuit_failure}) and the nearest attack strategy (second row of Figure~\ref{fig:pursuit_failure}), the prey school exhibits high cohesion and synchronized motion to evade predation. The school dynamically adjusts its position and demonstrates polarized movements in response to approaching predators, consistently maintaining separation.

This observed pattern, characterized by the prey school's synchronized and polarized evasive maneuvers, aligns closely with the findings of Katz et al. \cite{Katz2011}, who demonstrated that collective interactions within fish schools are crucial for maintaining group stability and effectively evading predation pressure.

\begin{figure}[H]
    \centering

    \includegraphics[scale=0.2]{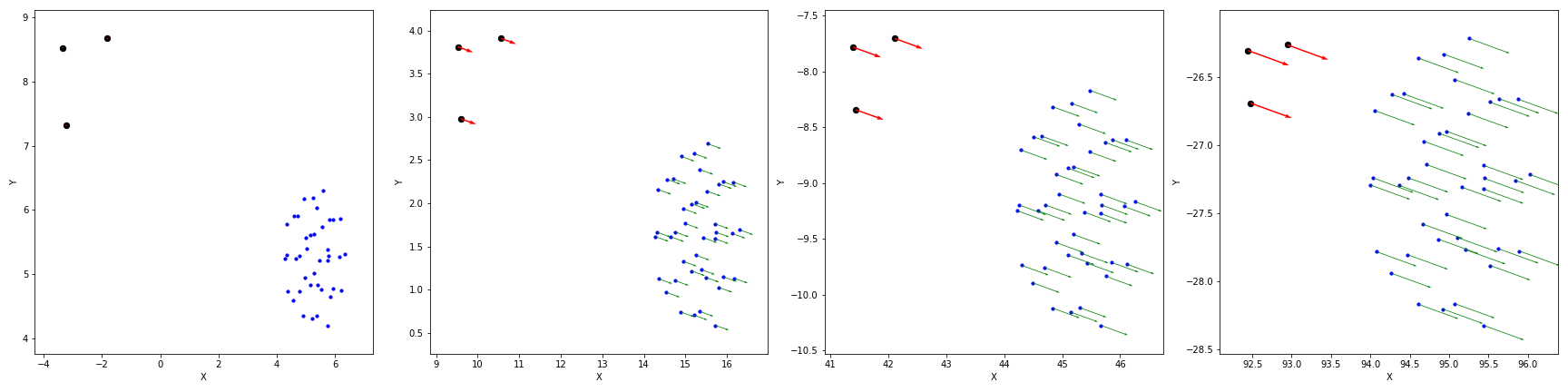}
    \caption*{Center attack strategy. Time steps (left to right): 
    $t = 0$, $1000$, $1999$, $2999$.}
    \label{fig:center_attack}

    \vspace{0.2cm} 

    \includegraphics[scale=0.2]{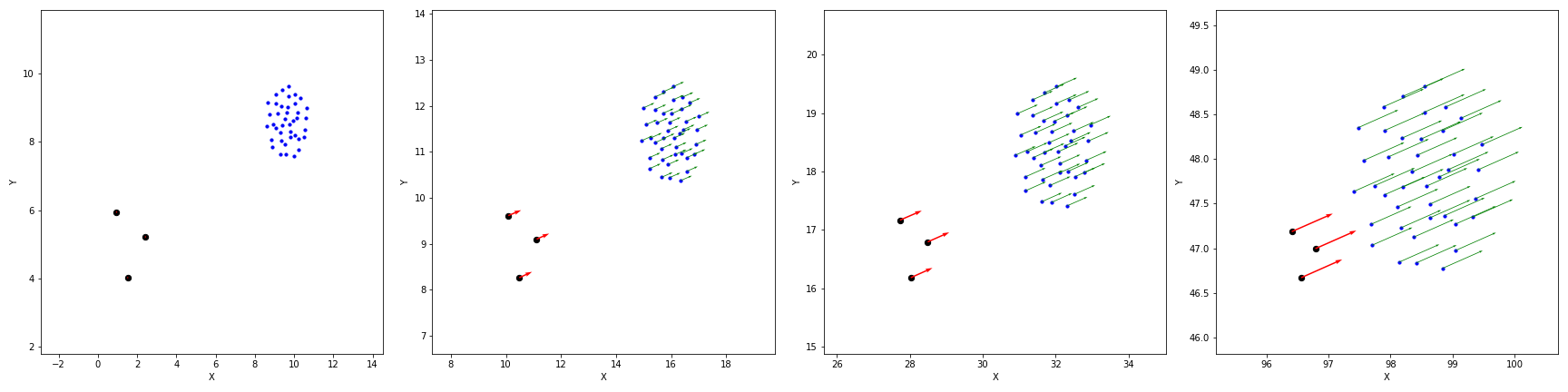}
\caption*{Nearest attack strategy. Time steps (left to right): 
    $t = 0$, $800$, $1500$, $2999$.}
    \label{fig:nearest_attack}

    \caption{Pattern 1: Pursuit failure under two predation strategies. 
    The predator school (black dots) fails to capture the prey school (blue dots) within the allotted timeframe.}
    \label{fig:pursuit_failure}
\end{figure}

\noindent\textbf{Pattern 2: Dispersal by Predators.}

This pattern (see Figure~\ref{fig:center_chongsan}) describes how predators, through collective movement, disrupt the cohesion of the prey school. Under the center attack strategy, predators focus on approaching the center of the prey school. The concentrated pressure from predators forces the prey school to break its original cohesive structure, causing prey individuals to disperse in multiple directions to evade predation.

This pattern is similar to that found in \cite{Pitcher1983}. Their study describes how predation pressure can lead to the fragmentation and dispersal of prey schools. Furthermore, the similar dispersal behavior observed in the prey school is closely related to the study by \cite{Couzin2005}, which highlight the role of decision-making and dispersal strategies in reducing predation risk.

\begin{figure}[H]
    \centering
    \includegraphics[scale=0.22]{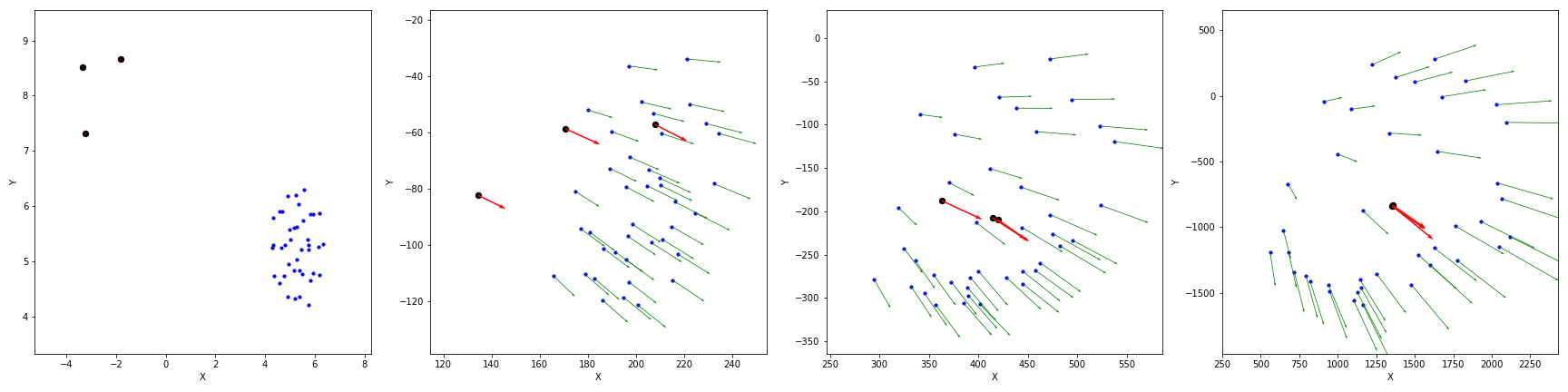}
   \caption{Pattern 2: Dispersal by predators under the center attack strategy. Predators (black dots) disrupt the cohesion of the prey school (blue dots), forcing individuals to scatter. Time steps (left to right): 
    $t = 0$, $650$, $1100$, $2999$.}
    \label{fig:center_chongsan}
\end{figure}

\noindent\textbf{Pattern 3: Selective Pursuit After Dispersal.} 

This pattern (see Figure~\ref{fig:pattern3}) describes the predator behavior following the successful dispersal of the prey school. Under the nearest attack strategy, predators move from collective hunting to individual hunting, with each predator independently selecting and pursuing the nearest prey. 

This pattern is consistent with Hamilton's selfish herd theory \cite{Hamilton1971}. The study by Pitcher and Parrish \cite{Pitcher1993} describes the importance of cohesion in defending against predation pressure, with dispersal significantly weakening defenses. Finally, Couzin et al. \cite{Couzin2005} explored how predators optimize their hunting strategies, particularly by implementing targeted attacks after prey dispersal.

\begin{figure}[H]
    \centering
    \includegraphics[scale=0.26]{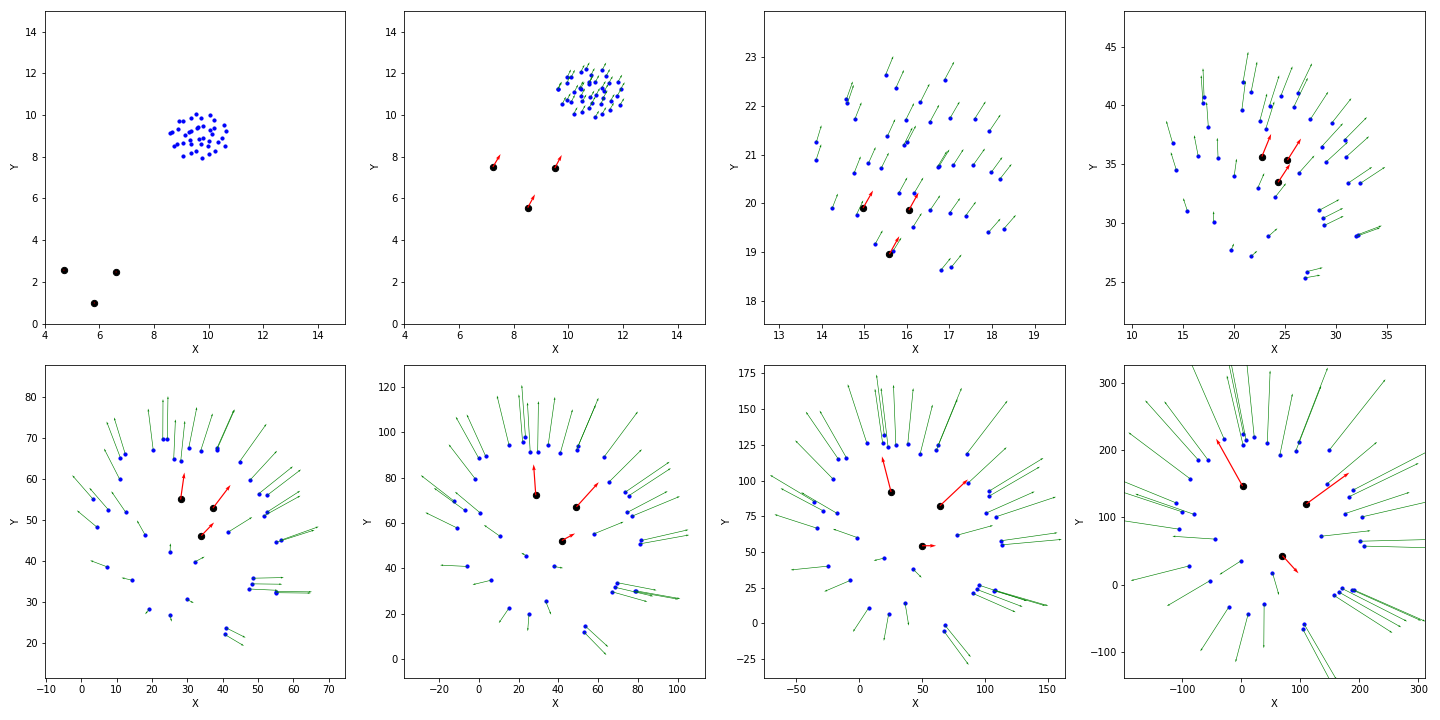}
   \caption{Pattern 3: Selective pursuit after dispersal under the nearest attack strategy. After dispersing the prey school (blue dots), predators (black dots) pursue the nearest prey individually, resulting in decentralized hunting. Time steps (left to right, top to bottom): 
    $t = 0$, $100$, $240$, $370$, $500$, $600$, $700$, $920$.}
    \label{fig:pattern3}
\end{figure}

\noindent\textbf{Pattern 4:  Fragmentation and Regrouping.} 

This pattern (see Figure~\ref{fig:pattern4_fragmentation}) represents a dynamic interaction under the center attack strategy. The predators disrupt the prey school, causing it to fragment into smaller clusters. After this disruption, predators adapt their behavior by dispersing to hunt these smaller prey clusters individually. In the final phase, predators regroup to target a single remaining prey cluster. Studies on fish shoaling behavior have shown that fragmentation often occurs when predators apply concentrated pressure to prey schools \cite{Pitcher1983, Brierley2010}. Such situations and dynamics are similar to those found in these studies \cite{Couzin2003, Freon1999, Herbert-Read2011}.

\begin{figure}[H]
    \centering
    \includegraphics[scale=0.22]{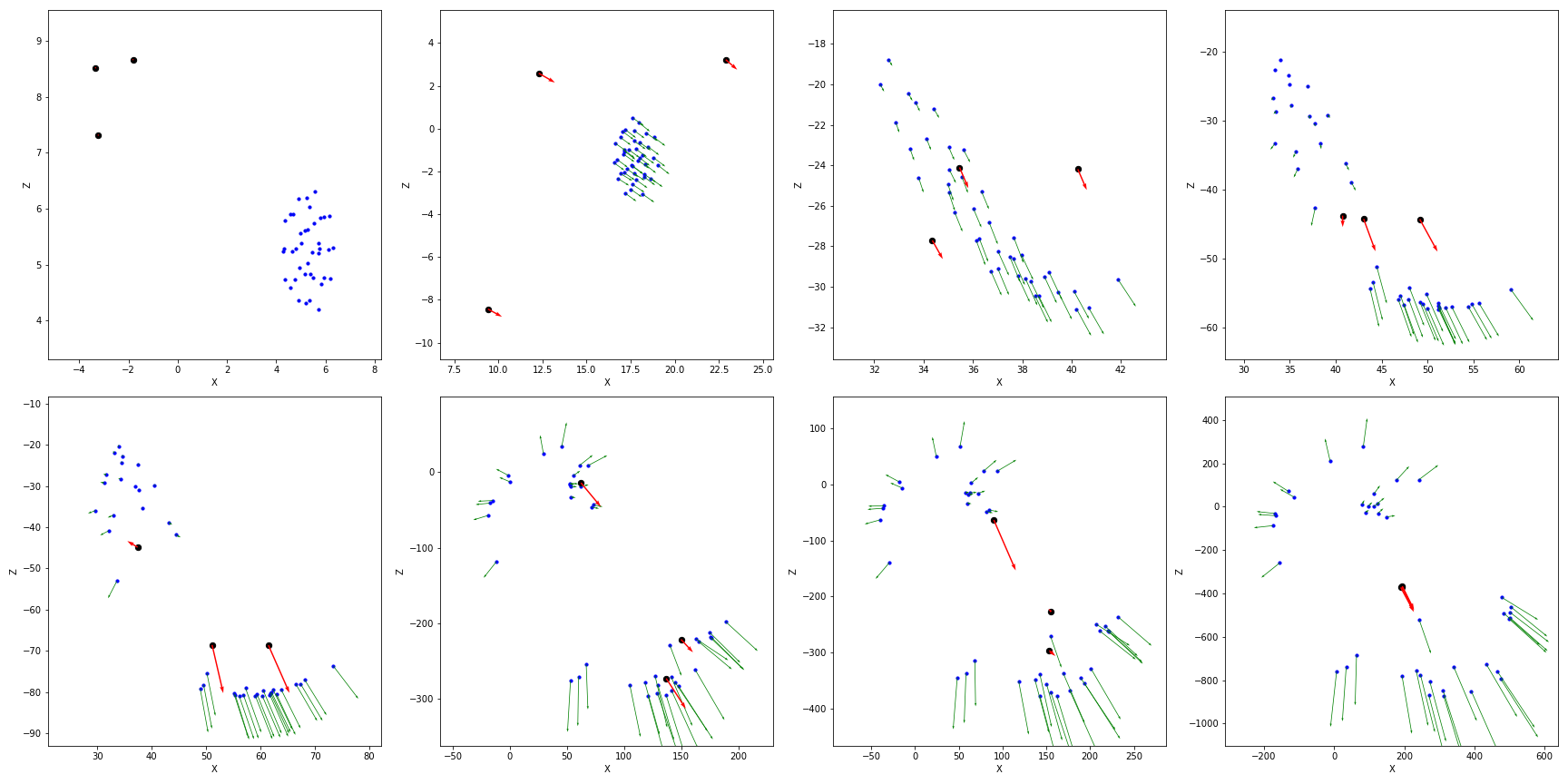}
   \caption{Pattern 4: Fragmentation and regrouping. Predators (black dots) fragment the prey school (blue dots), then engage in dispersed hunting before regrouping to target a single cluster. Time steps (left to right, top to bottom): 
    $t = 0$, $300$, $600$, $740$, $840$, $1500$, $1720$, $2999$.}
    \label{fig:pattern4_fragmentation}
\end{figure}

\noindent\textbf{Pattern 5: Edge Dispersal and Selfish Escape.}  

Under the nearest attack strategy, predators focus their attention on the nearest prey individuals at the edges of the prey school. This caused higher predation pressure on those individuals. As a result, edge prey may respond with a behavior known as selfish escape, where individuals break away from the prey school to evade predation (see Figure~\ref{fig:pattern5_selfish_escape}). 

The dynamics of this pattern are rooted in the selfish herd theory proposed by Hamilton  \cite{Hamilton1971}, which suggests that individuals in a group compete for safer positions, such as the center, to reduce predation risk. For prey at the edges, the intense pressure may make joining the core infeasible, leaving escape as the optimal strategy.

\begin{figure}[H]
    \centering
    \includegraphics[scale=0.24]{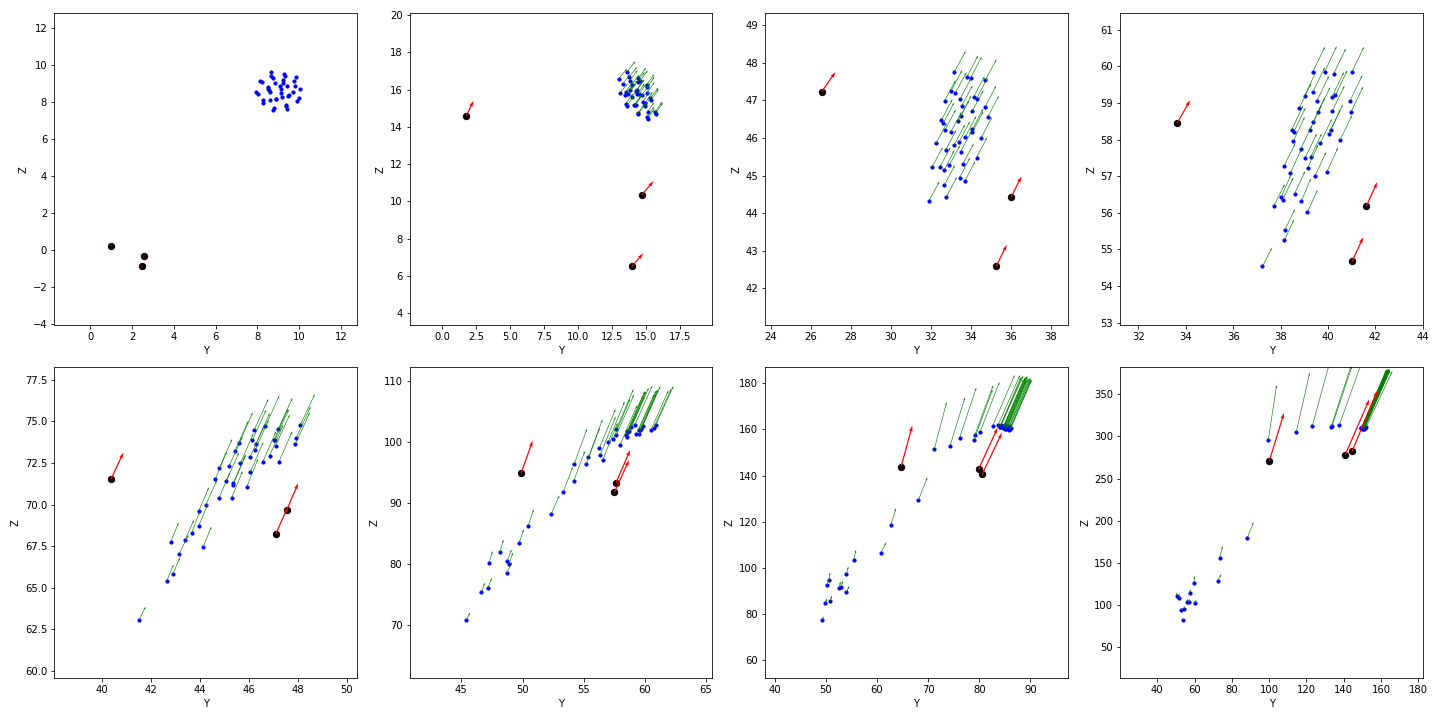}
   \caption{Pattern 5: Edge dispersal and selfish escape under the nearest attack strategy. Predators (black dots) exert pressure on edge prey, causing individuals to peel away from the school (blue dots) in selfish escape attempts. Time steps (left to right, top to bottom): 
    $t = 0$, $200$, $500$, $580$, $660$, $780$, $980$, $1400$.}
    \label{fig:pattern5_selfish_escape}
\end{figure}

\noindent\textbf{Pattern 6: Oscillatory Predation with Cohesive Defense.}  
This pattern describes a high-pressure dynamic where the predators engage in repeated, oscillatory attacks while the prey school's cohesion and structural stability remain robust. This is observed across both the center attack (first panel of Figure~\ref{fig:oscillatory_dynamics}) and nearest attack (last two rows of Figure~\ref{fig:oscillatory_dynamics}) strategies.

The predators execute a strategy of oscillatory predation, repeatedly adjusting their attack vectors in an attempt to disorient and penetrate the formation—a dynamic observed in natural events such as the bait-ball attacks by sharks and tuna, as documented by Pitcher \cite{Pitcher1983}. Despite this intense, multi-directional pressure, the prey school's synchronized movements allow it to maintain its structural integrity.

The persistent cohesion under active attack supports both the Selfish Herd Theory \cite{Hamilton1971}, as individuals reduce risk by remaining within the group core, and the conclusions of Couzin et al. \cite{Couzin2005}, who identified synchronized movement as a critical defense mechanism that allows groups to withstand active predation without compromising overall structure.

\begin{figure}[H]
    \centering

    \includegraphics[scale=0.2]{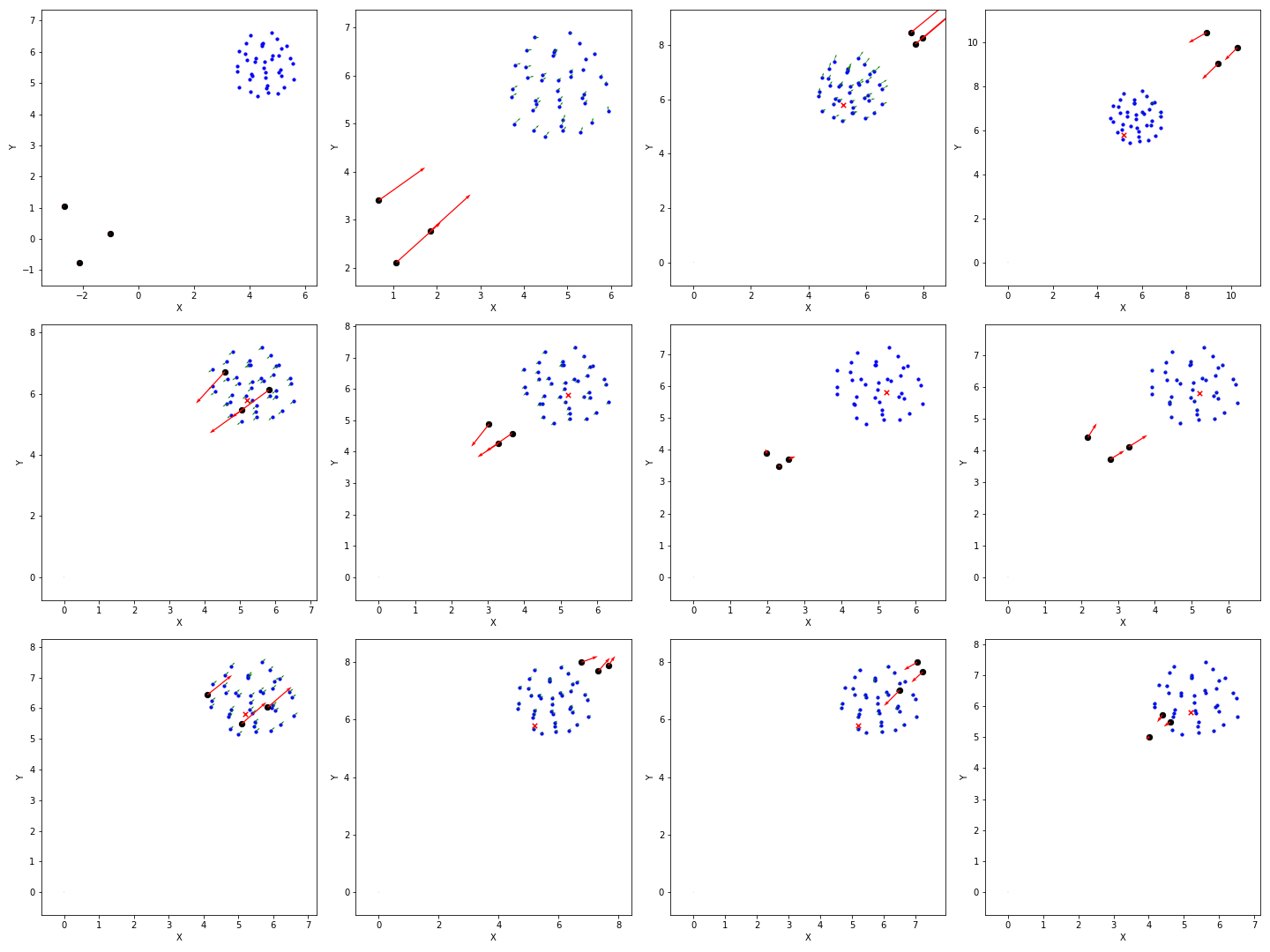}
    \caption*{Center attack strategy. Time steps (left to right, top to bottom): 
    $t = 0$, $120$, $250$, $410$, $550$, $610$, $710$, $760$, $860$, $965$, $1120$, $1300$.}
    \label{fig:center_attack_oscillation}
    
    \vspace{0.25cm} 

    \includegraphics[scale=0.2]{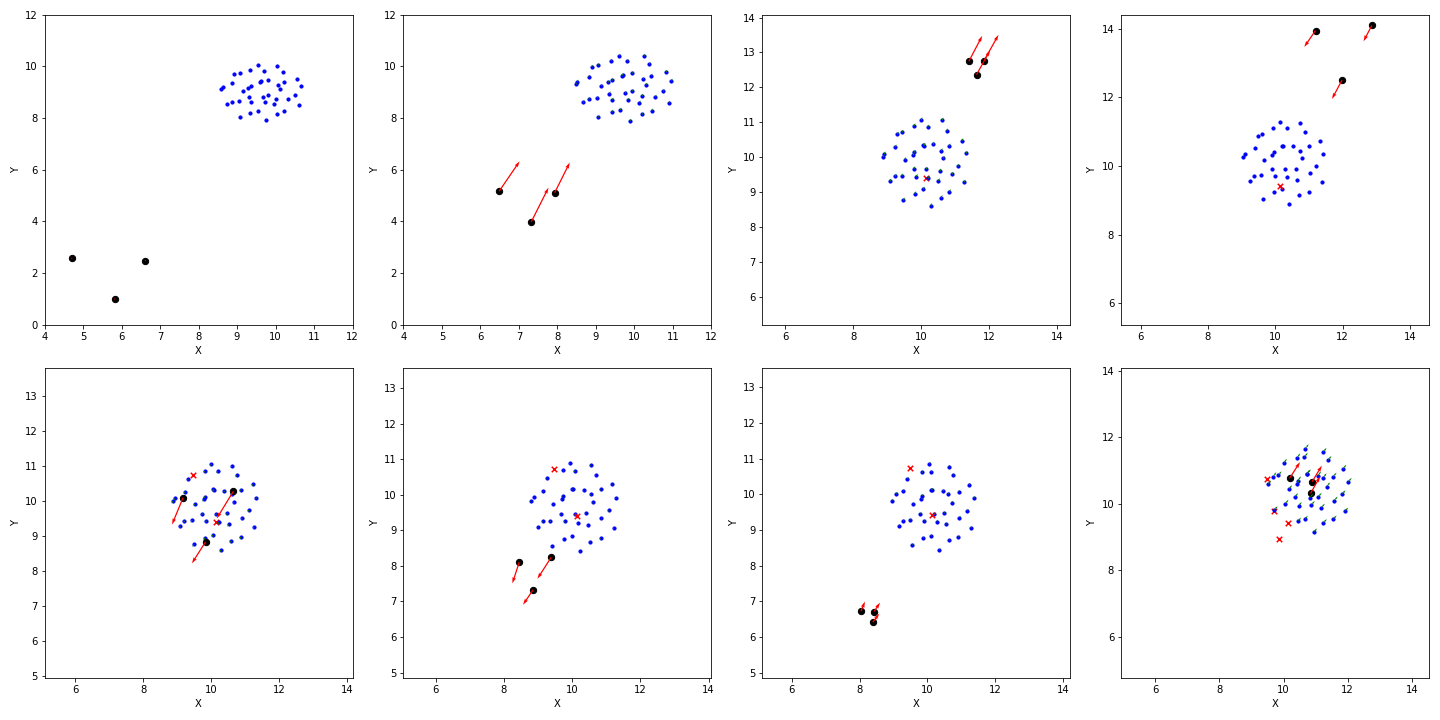}
    \caption*{Nearest attack strategy. Time steps (left to right, top to bottom): 
    $t = 0$, $100$, $240$, $400$, $500$, $550$, $700$, $850$.}
    \label{fig:nearest_attack_oscillation}
    
   \caption{Pattern 6: Oscillatory predation dynamics under center and nearest attack strategies. 
    The prey school (blue dots) resists disruption while predators (black dots) apply persistent pressure. 
    Red crosses mark prey captured during each event.}
    \label{fig:oscillatory_dynamics}
\end{figure}

\noindent\textbf{Pattern 7: Predator Surrounding.}

The pattern is characterized by the predators forming a stable, encircling structure around the prey school, a behavior observed under both the center attack (first panel) and nearest attack (second panel) strategies (Figure~\ref{fig:surrounding_behavior}).

In this pattern, the predators coordinate their movements to establish and maintain a tight, surrounding formation, thereby applying significant spatial pressure and restricting the prey school's potential escape vectors. In response, the prey demonstrate a high degree of strong cohesion and stable formation, effectively utilizing this tight group structure to resist the encompassing pressure and delay capture.

This surrounding phenomenon is a common cooperative predation strategy in natural marine ecosystems. For instance, Krause and Ruxton \cite{Krause2002} observed that predators like sharks and tuna employ coordinated, cooperative behaviors specifically to encircle prey schools, thereby eliminating their escape options. The prey school's successful resistance through tight cohesion reinforces findings by Couzin et al. \cite{Couzin2005}, which demonstrated that group cohesion and synchronized movements are critical mechanisms used by prey to mitigate individual exposure when faced with external spatial constraints.

\begin{figure}[H]
    \centering

    \includegraphics[scale=0.2]{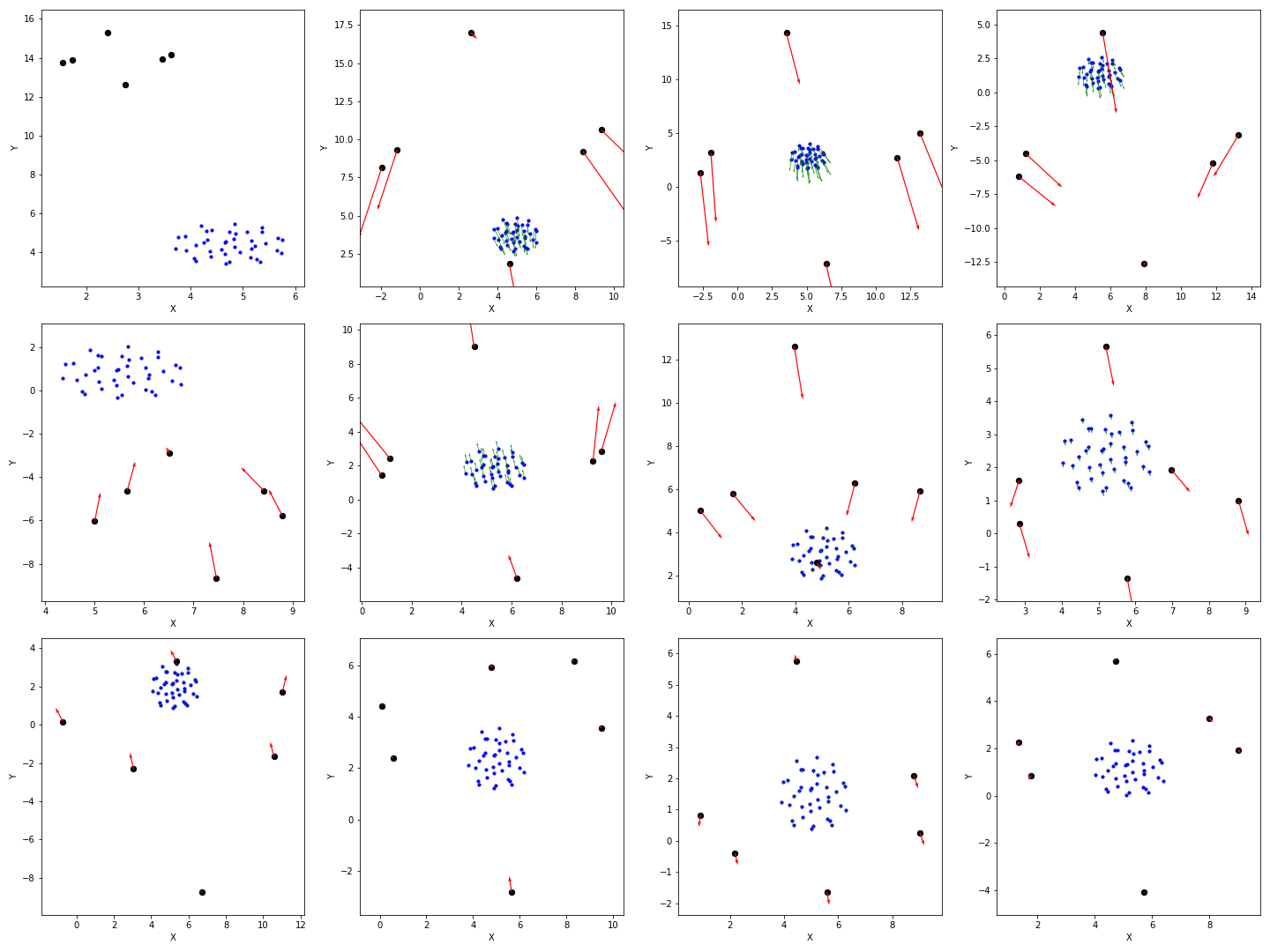}
   \caption*{Center attack strategy. Time steps (left to right, top to bottom): 
    $t = 0$, $120$, $200$, $320$, $450$, $600$, $880$, $1000$, $1200$, $2200$, $2500$, $2999$.}
    \label{fig:center_surrounding}
    
    \vspace{0.25cm} 

    \includegraphics[scale=0.2]{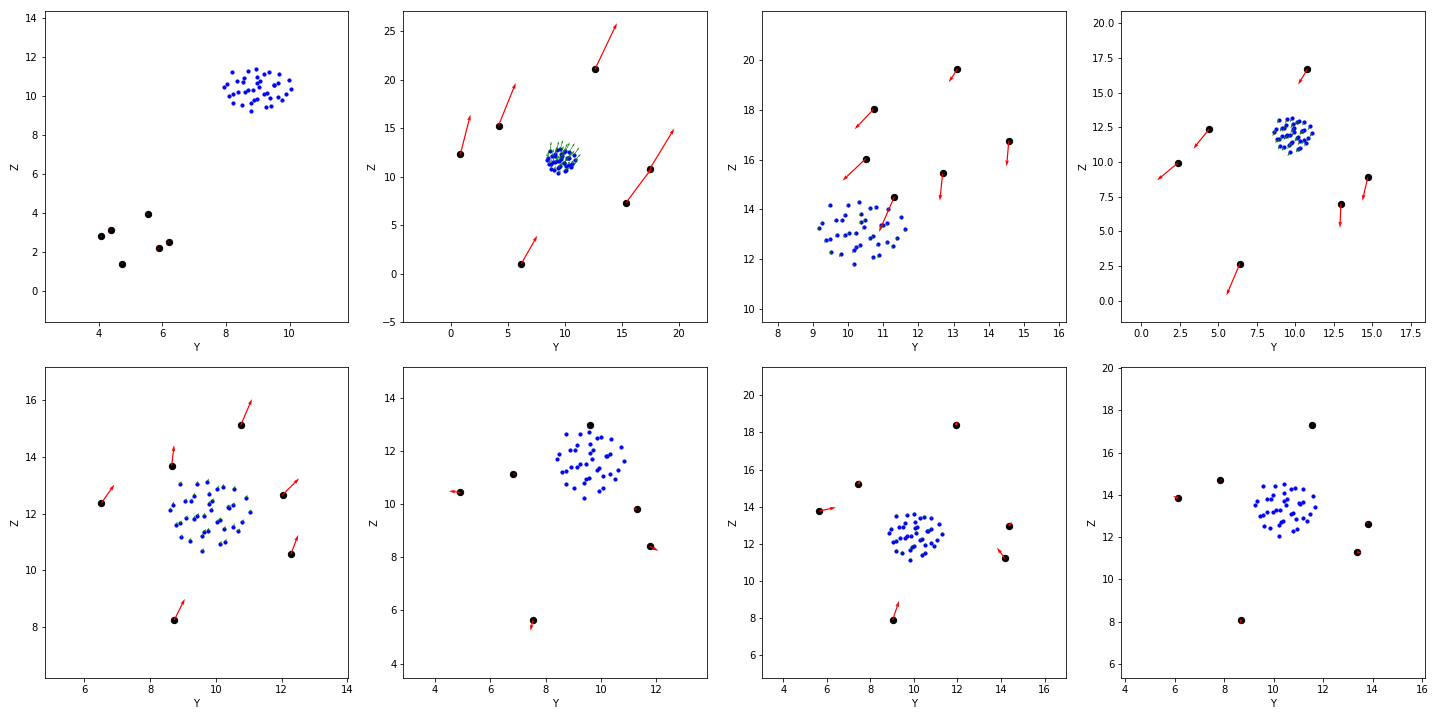}
   \caption*{Nearest attack strategy. Time steps (left to right, top to bottom): 
    $t = 0$, $200$, $500$, $660$, $980$, $1500$, $1800$, $2999$.}
    \label{fig:nearest_surrounding}

 \caption{Pattern 7: Predator surrounding under center and nearest attack strategies. 
    Predators (black dots) form a stable encircling formation around the prey school (blue dots), which resists capture through cohesion and stability.}
    \label{fig:surrounding_behavior}
\end{figure}

\section{Impact of predator school size on prey dynamics}
\label{Effectsofpredatorschoolsize}

In this section, we present  the second set of simulation results, focusing on how predator school size influences prey school dynamics. 

We fix the prey population at \( N = 300 \) and vary the number of predators from \( M = 1 \) to \( M = 20 \). For each predator school size, we perform 100 independent simulation trials and report average outcomes. 

To initialize the prey school, we first solve the prey-only system (without predators) until a cohesive group of \( N = 300 \) individuals emerges, with all prey positions stored. In each trial, a predator school of size \( M \) is then generated in the same manner. Predator positions vary across trials, introducing stochastic variability in addition to inherent system noise. This randomness reflects natural variability in real aquatic environments.  

To minimize bias from initial predator placement, we rotate the predator school (while preserving relative positions) until the standard deviation of predator-to-prey centroid distances is minimized. The two centroids are then placed 12 units apart. The pseudocode for this initialization procedure is provided in Algorithm~\ref{alg:scheme4}.

\begin{algorithm}[H]
\caption{Initialization of prey and predator positions for simulation}
\label{alg:scheme4}
\textbf{Input:} $N$: number of prey, $M$: number of predators, $x_{ps}, y_{pdts}$: initial positions of prey and predators, $\text{desired\_distance}$: target distance between predator and prey centers. \\
\textbf{Output:} $x_0$: prey initial positions, $y_0$: predator initial positions adjusted to maintain balanced distances from prey center.
\begin{algorithmic}[1]
\State $x_0 \gets x_{ps}[-1]$ \Comment{Initialize prey positions}
\State $v_0 \gets \mathbf{0}$ \Comment{Initialize prey velocities}
\State $c_x \gets \text{mean}(x_0, \text{axis}=0)$ \Comment{Prey center}

\State $y_{\text{init}} \gets y_{pdts}[-1]$ \Comment{Initialize predator positions}
\State $c_y \gets \text{mean}(y_{\text{init}}, \text{axis}=0)$ \Comment{Predator center}
\State $r \gets y_{\text{init}} - c_y$ \Comment{Relative predator positions}

\State $\text{dir} \gets \dfrac{c_x - c_y}{\|c_x - c_y\|}$ \Comment{Direction from predator center to prey center}
\State $\text{new\_center} \gets c_x - \text{dir}\cdot \text{desired\_distance}$ \Comment{Shift predator center}
\State $y_0 \gets \text{new\_center} + r$ \Comment{Adjusted predator positions}

\State \textbf{assert} $\;\text{allclose}(r,\;y_0 - \text{mean}(y_0,\text{axis}=0))$ \Comment{Check relative positions preserved}

\Function{RotateToMinimizeDistanceVariance}{$y_0, c_x$}
    \State $d \gets \|y_0 - c_x\|$; $best\_std \gets \text{std}(d)$; $best\_y \gets y_0$
    \For{$\theta$ in evenly spaced angles on $[0,2\pi]$}
        \State $R \gets \begin{bmatrix}
        \cos\theta & -\sin\theta & 0 \\
        \sin\theta & \cos\theta & 0 \\
        0 & 0 & 1
        \end{bmatrix}$
        \State $y_{\text{rot}} \gets \text{new\_center} + R \cdot r$
        \State $d_{\text{rot}} \gets \|y_{\text{rot}} - c_x\|$
        \If{$\text{std}(d_{\text{rot}}) < best\_std$}
            \State $best\_std \gets \text{std}(d_{\text{rot}})$; $best\_y \gets y_{\text{rot}}$
        \EndIf
    \EndFor
    \State \Return $best\_y,\; best\_std$
\EndFunction

\State $y_0, final\_std \gets \textsc{RotateToMinimizeDistanceVariance}(y_0, c_x)$
\State \Return $x_0, y_0$
\end{algorithmic}
\end{algorithm}

All simulations use the following parameters: 
\[
\alpha = 10,\; \beta = 0.1,\; r = 1,\; p = 4,\; q = 5,\; 
\gamma_1 = 1,\; \gamma_2 = 0.1,\; R = 2,\; R_1 = 2,\; R_2 = 2,\;
\theta = 2,\]
\[
 \theta_1 = 1,\; \theta_2 = 0.3,\; \delta = 0.1,
\]
with a time step of \( \Delta t = 0.01 \) and total simulation time of 40 units (4000 steps). Parallel computing accelerates the process by approximately tenfold (Algorithm~\ref{alg:scheme5}). 

We analyze the following metrics to evaluate predation efficiency and prey outcomes:
\begin{itemize}
    \item \textbf{Average Total Prey Captured:} This represents the average total number of prey captured by all predators across all simulation trails:
    \[
    \bar{F} = \frac{1}{N_\text{sim}} \sum_{i=1}^{N_\text{sim}} F_i
    \]
    where $N_\text{sim}$ is the total number of simulation trails, and $F_i$ is the total number of prey captured in the $i$-th simulation trail.

    \item \textbf{Average Max–Min Difference:} This indicates the average difference between individuals in a predator school capturing the most prey and capturing the least prey:
    \[
    \overline{\Delta F} = \frac{1}{N_\text{sim}} \sum_{i=1}^{N_\text{sim}} \left( \max(P_i) - \min(P_i) \right)
    \]
    where $P_i$ is the set of prey capture counts for all predators in the $i$-th simulation trail.

    \item \textbf{Average Standard Deviation:} This quantifies the variation in prey capture counts among predators, averaged over all simulation trails:
    \[
    \bar{\sigma} = \frac{1}{N_\text{sim}} \sum_{i=1}^{N_\text{sim}} \sigma(P_i)
    \]
    where $\sigma(P_i)$ is the standard deviation of prey capture counts in the $i$-th simulation trail.

    \item \textbf{Average Prey Captured per Predator:} This represents the average number of prey captured by a single predator across all simulation trails:
    \[
    \bar{\mu} = \frac{1}{N_\text{sim}} \sum_{i=1}^{N_\text{sim}} \mu(P_i)
    \]
    where $\mu(P_i)$ is the mean prey capture count for all predators in the $i$-th simulation trail.

    \item \textbf{Average Prey Survival Time:} This reflects the average survival time of prey during the simulations:
    \[
    \bar{T}_\text{alive} = \frac{1}{N_\text{sim}} \sum_{i=1}^{N_\text{sim}} T_i
    \]
    where $T_i$ is the average survival time of prey in the $i$-th simulation trail.

    \item \textbf{Average Prey Survival Rate:} This provides the average survival rate of prey across all simulations:
    \[
    \bar{R} = \frac{1}{N_\text{sim}} \sum_{i=1}^{N_\text{sim}} R_i
    \]
    where $R_i$ is the survival rate of prey in the $i$-th simulation trail.
\end{itemize}

\begin{algorithm}[H]
\caption{Parallel simulation of predator--prey dynamics and statistical analysis}
\label{alg:scheme5}
\textbf{Input:} $N, M$: number of prey and predators; $t$: time vector; $x_0, v_0$: initial prey positions and velocities; \\
\hspace{20pt} Simulation parameters: $\alpha, \beta, r, p, q, \gamma, R, \theta$, etc.; \\
\hspace{20pt} $num\_simulations$: total number of simulations. \\
\textbf{Output:} Summary metrics: $total\_eaten\_fish, dm, std, avg\_pred, survival\_rate, survival\_time$.
\begin{algorithmic}[1]
\State Initialize parameter sets: $noise\_levels$, $distances$, etc.
\State Define function \textsc{RunSimulationParallel} to compute prey--predator dynamics.
\Statex

\State \textbf{Run parallel simulations:}
\For{$sim = 1$ to $num\_simulations$}
    \State Simulate prey--predator interactions using $force\_function\_prey$ and $force\_function\_predator$.
    \State Record prey capture events and predator capture counts.
\EndFor
\Statex

\State \textbf{Process results:}
\State Initialize lists: $L_{eaten}, L_{dm}, L_{std}, L_{avg}, L_{survival\_time}, L_{survival\_rate}$.
\For{each simulation result}
    \State Extract $total\_eaten\_fish, pred\_count, survival\_time, survival\_rate$.
    \State $dm \gets \max(pred\_count) - \min(pred\_count)$
    \State $std \gets \text{std}(pred\_count)$
    \State $avg\_pred \gets \text{mean}(pred\_count)$
    \State Append metrics to corresponding lists.
\EndFor
\Statex

\State \textbf{Summarize results:}
\State Compute averages:
\State \hspace{1em} $overall\_fish\_eaten \gets \text{mean}(L_{eaten})$
\State \hspace{1em} $overall\_dm \gets \text{mean}(L_{dm})$
\State \hspace{1em} $overall\_std \gets \text{mean}(L_{std})$
\State \hspace{1em} $overall\_avg\_pred \gets \text{mean}(L_{avg})$
\State \hspace{1em} $overall\_survival\_time \gets \text{mean}(L_{survival\_time})$
\State \hspace{1em} $overall\_survival\_rate \gets \text{mean}(L_{survival\_rate})$

\State Print details of individual simulations and global summary.
\State \Return all processed results and overall summary metrics.
\end{algorithmic}
\end{algorithm}

\subsection{Nearest Attack Strategy}

Figure~\ref{fig:nearest_avg} shows the relationship between predator school size (\(M\)) and total prey captured. For small groups ($1 \leq M \leq 3$), prey captures rise sharply as predators cooperate with limited competition. As group size increases ($4 \leq M \leq 20$), growth slows and eventually levels off due to limited prey availability (\(N=300\)). Internal competition increasingly constrains total captures, indicating diminishing returns to larger predator schools.

\begin{figure}[H]
    \centering
    \includegraphics[scale=0.3]{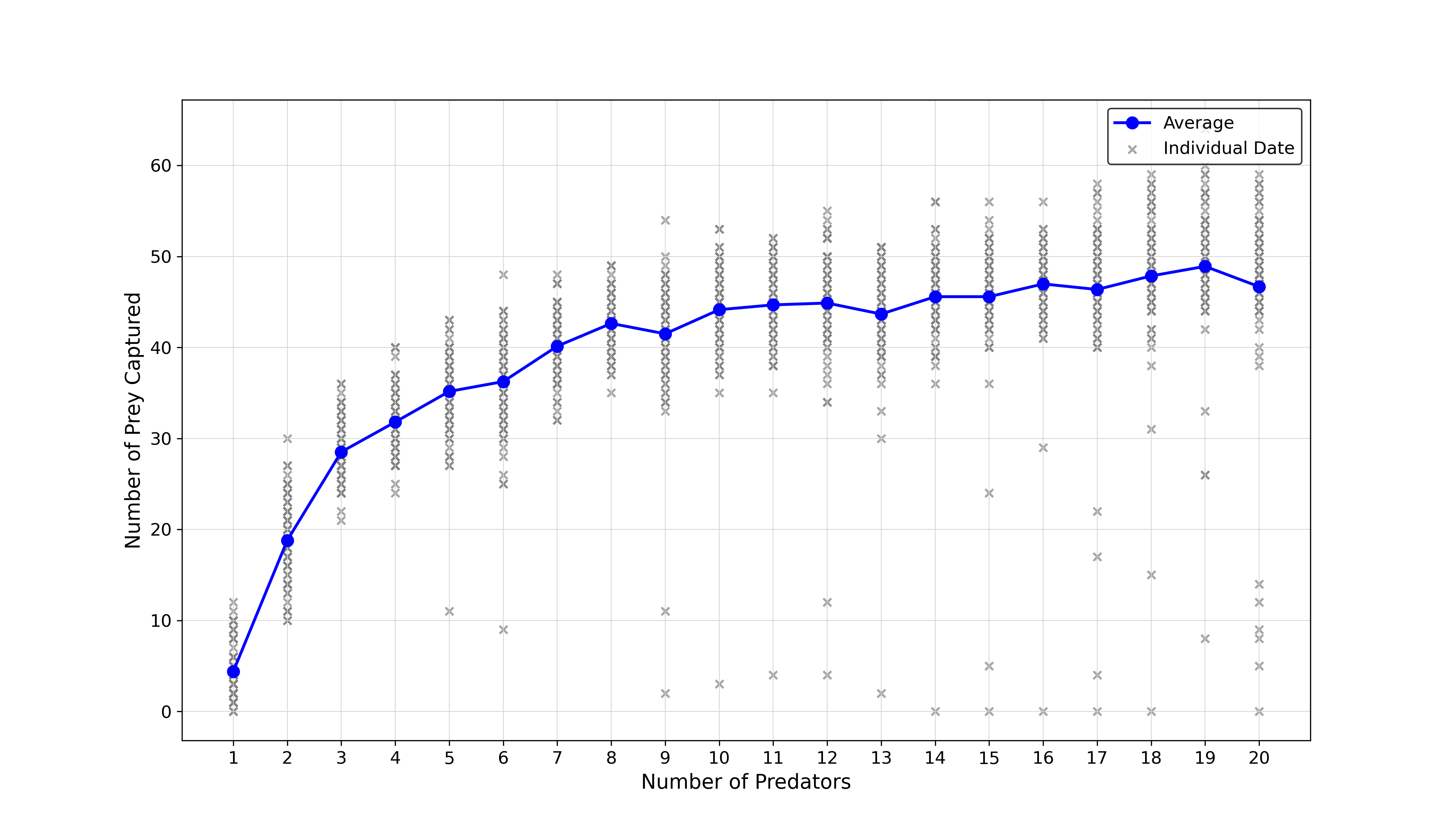}
    \caption{Expectation of Prey Captured by Predators under Nearest Attack Strategy}
    \label{fig:nearest_avg}
\end{figure}

Figure~\ref{fig:nearest_dm} shows the max–min difference in captures per predator. With few predators ($1 \leq M \leq 3$), differences rise rapidly, suggesting early competition. For larger groups, the difference stabilizes, as interference and limited prey balance predation opportunities across individuals. The standard deviation of captures (Figure~\ref{fig:nearest_std}) follows the same trend.

\begin{figure}[H]
    \centering
    \includegraphics[scale=0.3]{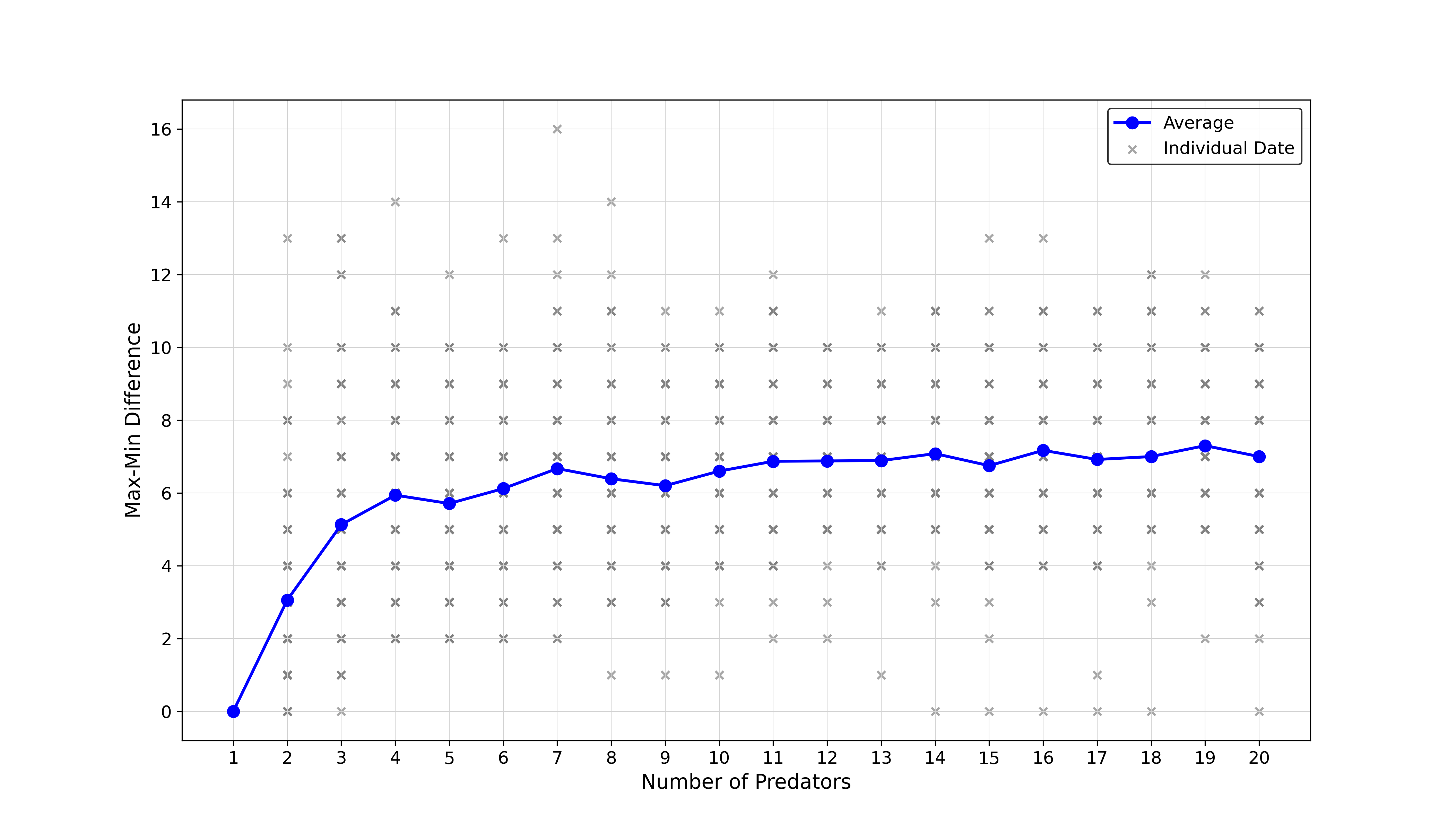}
    \caption{Max-Min Difference of Average Prey Captured by Predators under Nearest Attack Strategy}
    \label{fig:nearest_dm}
\end{figure}

\begin{figure}[H]
    \centering
    \includegraphics[scale=0.3]{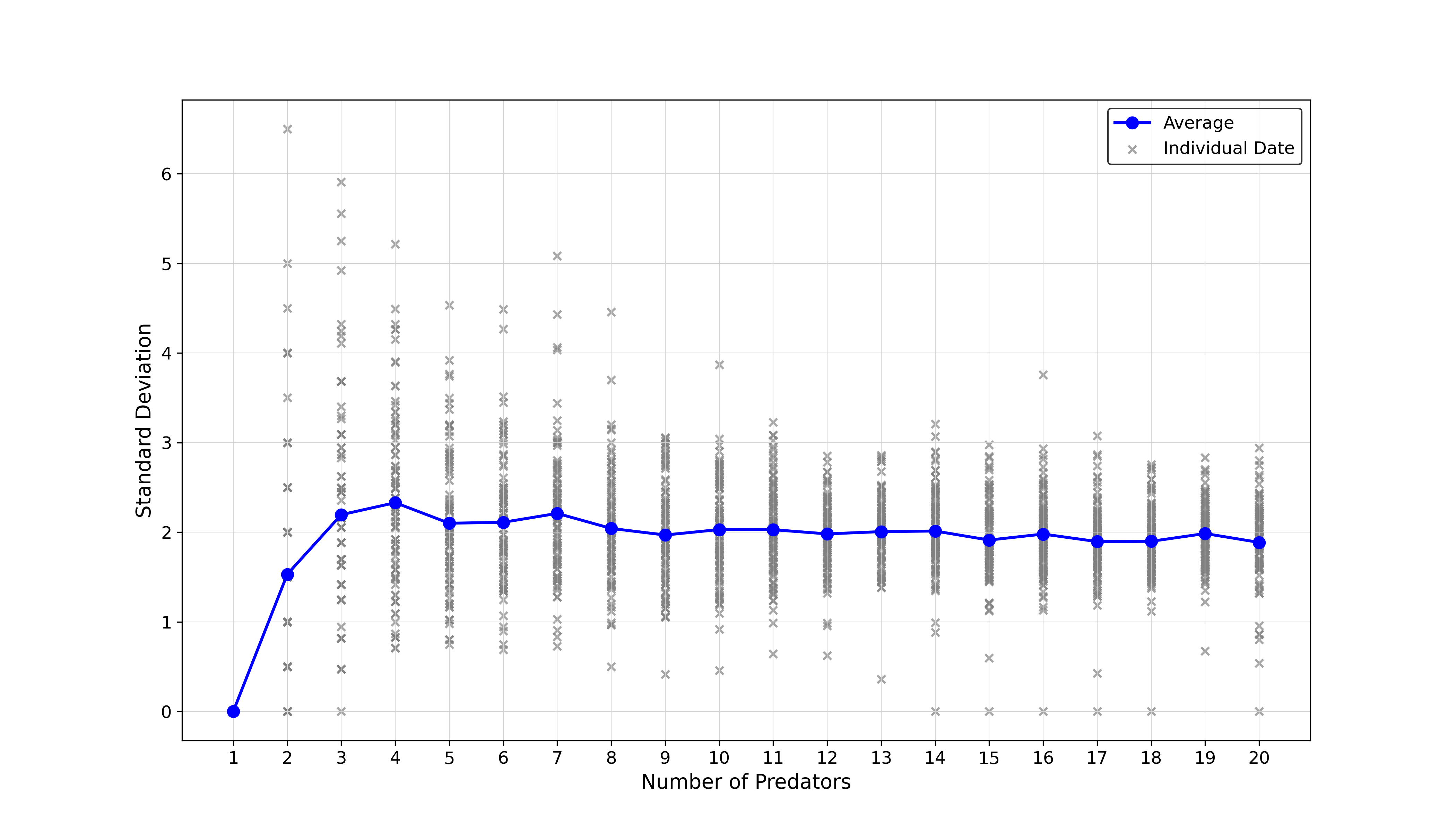}
    \caption{Standard Deviation of Average Prey Captured by Predators under Nearest Attack Strategy}
    \label{fig:nearest_std}
\end{figure}

The average prey captured per predator (Figure~\ref{fig:nearest_per}) rises briefly ($1 \leq M \leq 2$) due to cooperative hunting, but declines thereafter as competition dominates. Once group size exceeds 11, individual predators capture fewer prey on average than solitary hunters ($M=1$), indicating negative returns to cooperation. 

\begin{figure}[H]
    \centering
    \includegraphics[scale=0.3]{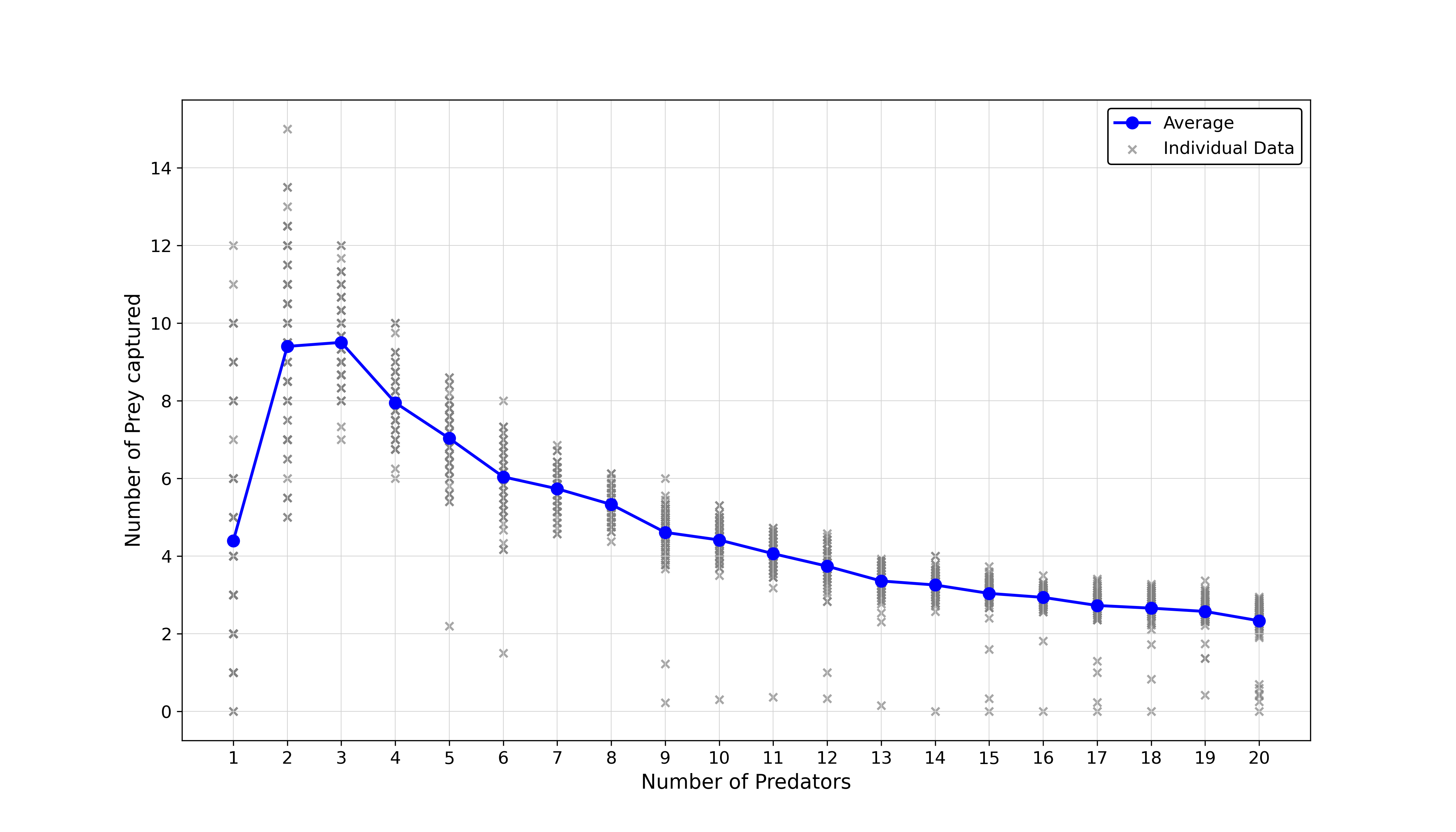}
    \caption{Average Number of Prey Captured per Predator under Nearest Attack Strategy}
    \label{fig:nearest_per}
\end{figure}

Figures~\ref{fig:nearest_time} and~\ref{fig:nearest_rate} show prey survival outcomes. Survival time and survival rate both decline steeply when predator numbers are small ($1 \leq M \leq 8$), but plateau for larger groups. Competition and interference reduce hunting efficiency, preventing further decreases.

\begin{figure}[H]
    \centering
    \includegraphics[scale=0.3]{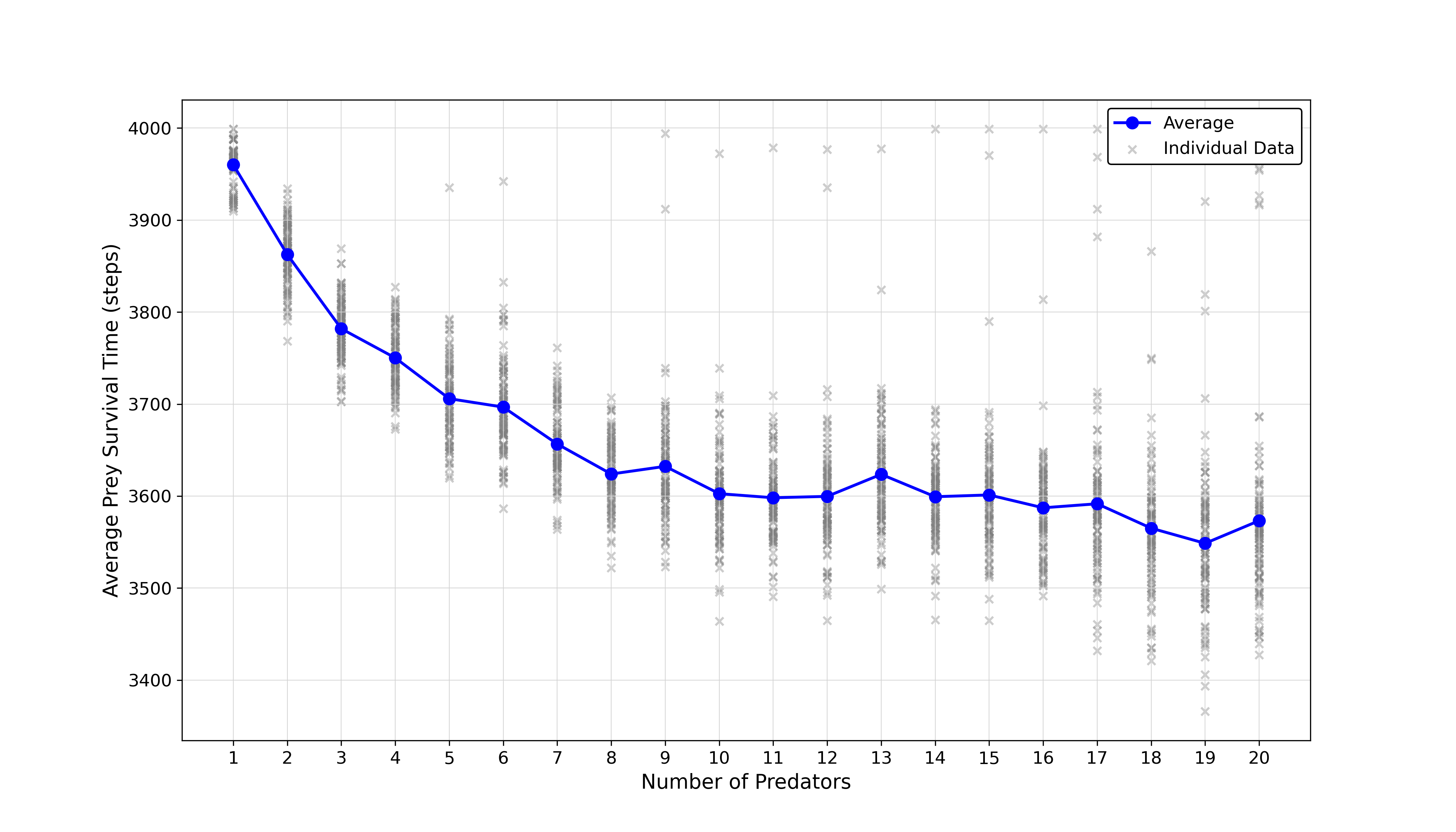}
    \caption{Prey Survival Time Across Different Predator Numbers under Nearest Attack Strategy}
    \label{fig:nearest_time}
\end{figure}

\begin{figure}[H]
    \centering
    \includegraphics[scale=0.35]{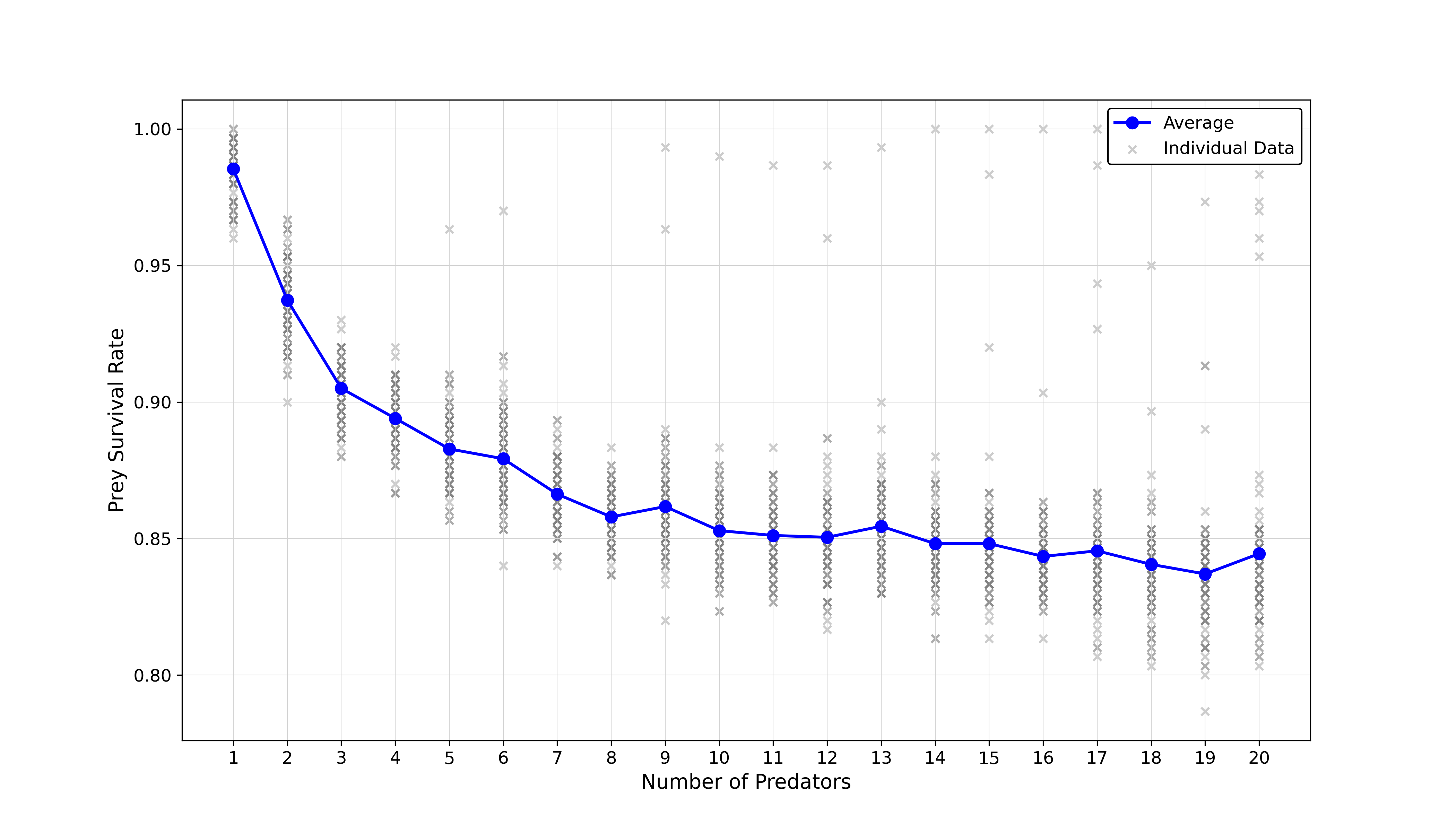}
    \caption{Prey Survival Rate Across Different Predator Numbers under Nearest Attack Strategy}
    \label{fig:nearest_rate}
\end{figure}

\subsection{Center Attack Strategy}

Figure~\ref{fig:center_avg} shows total prey captured under the center attack strategy. Captures increase at first ($1 \leq M \leq 3$), but growth slows ($4 \leq M \leq 12$) and then reverses beyond $M=12$. Concentration of predators at the prey center, combined with repulsive interactions, generates interference and lowers efficiency. In some cases, prey dispersal further reduces captures. 

\begin{figure}[H]
    \centering
    \includegraphics[scale=0.3]{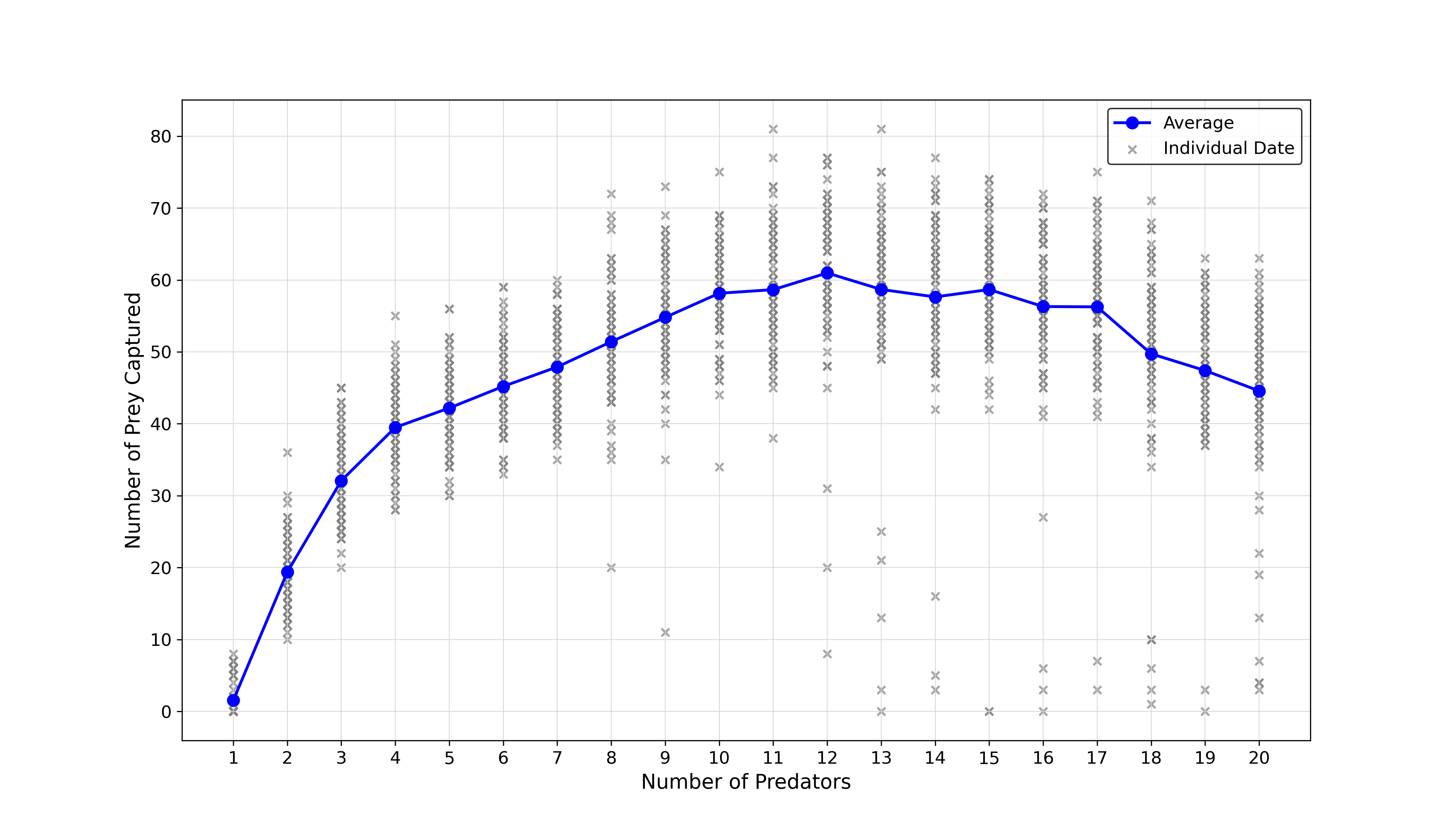}
    \caption{Expectation of Prey Captured by Predators under Center Attack Strategy}
    \label{fig:center_avg}
\end{figure}

Figures~\ref{fig:center_dm} and~\ref{fig:center_std} show that variation in captures (max–min difference and standard deviation) rises initially ($1 \leq M \leq 3$), stabilizes for moderate group sizes, and then decreases as strong competition reduces the advantage of more successful predators.

\begin{figure}[H]
    \centering
    \includegraphics[scale=0.3]{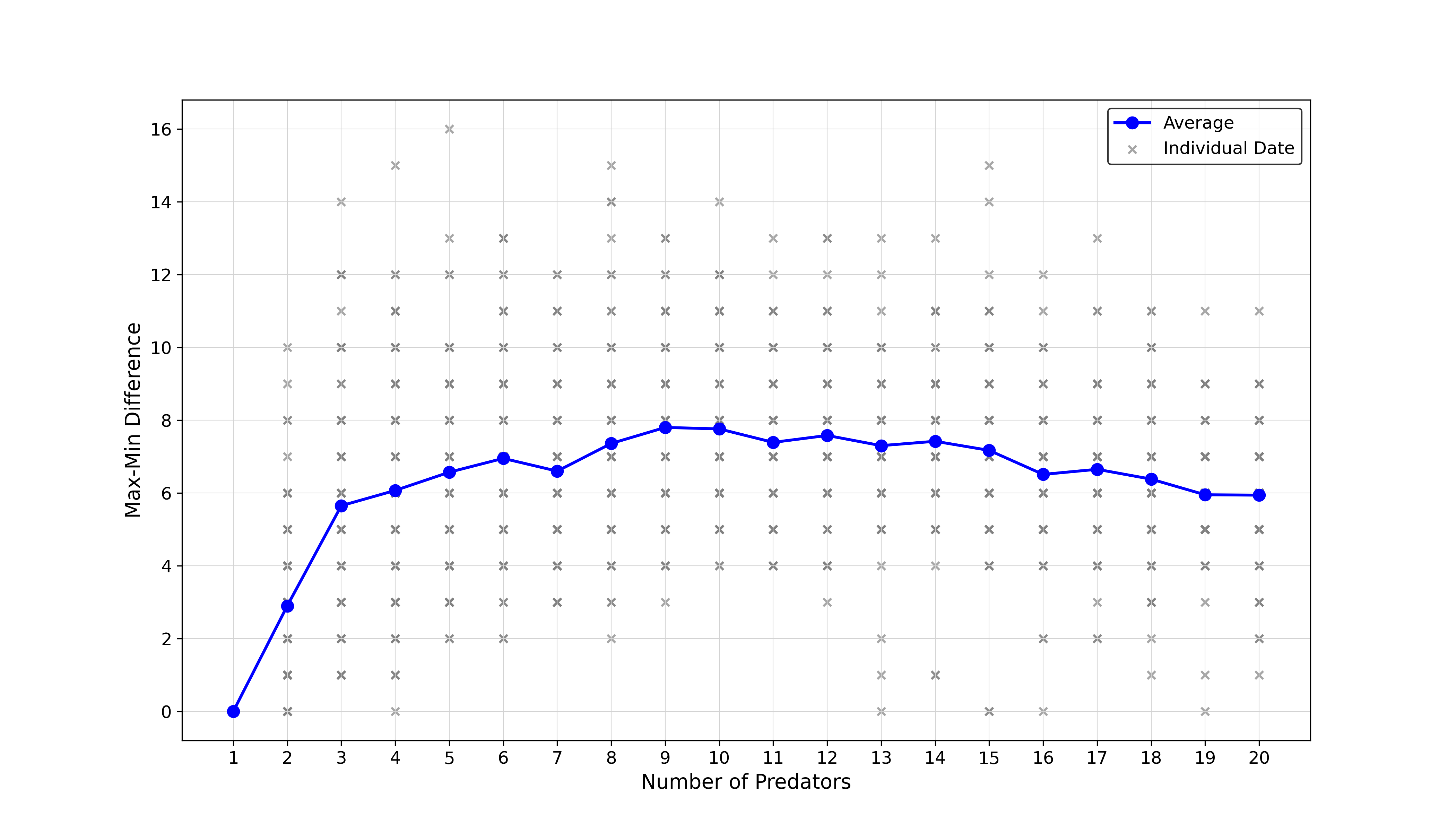}
    \caption{Max-Min Difference of Average Prey Captured by Predators under Center Attack Strategy}
    \label{fig:center_dm}
\end{figure}

\begin{figure}[H]
    \centering
    \includegraphics[scale=0.3]{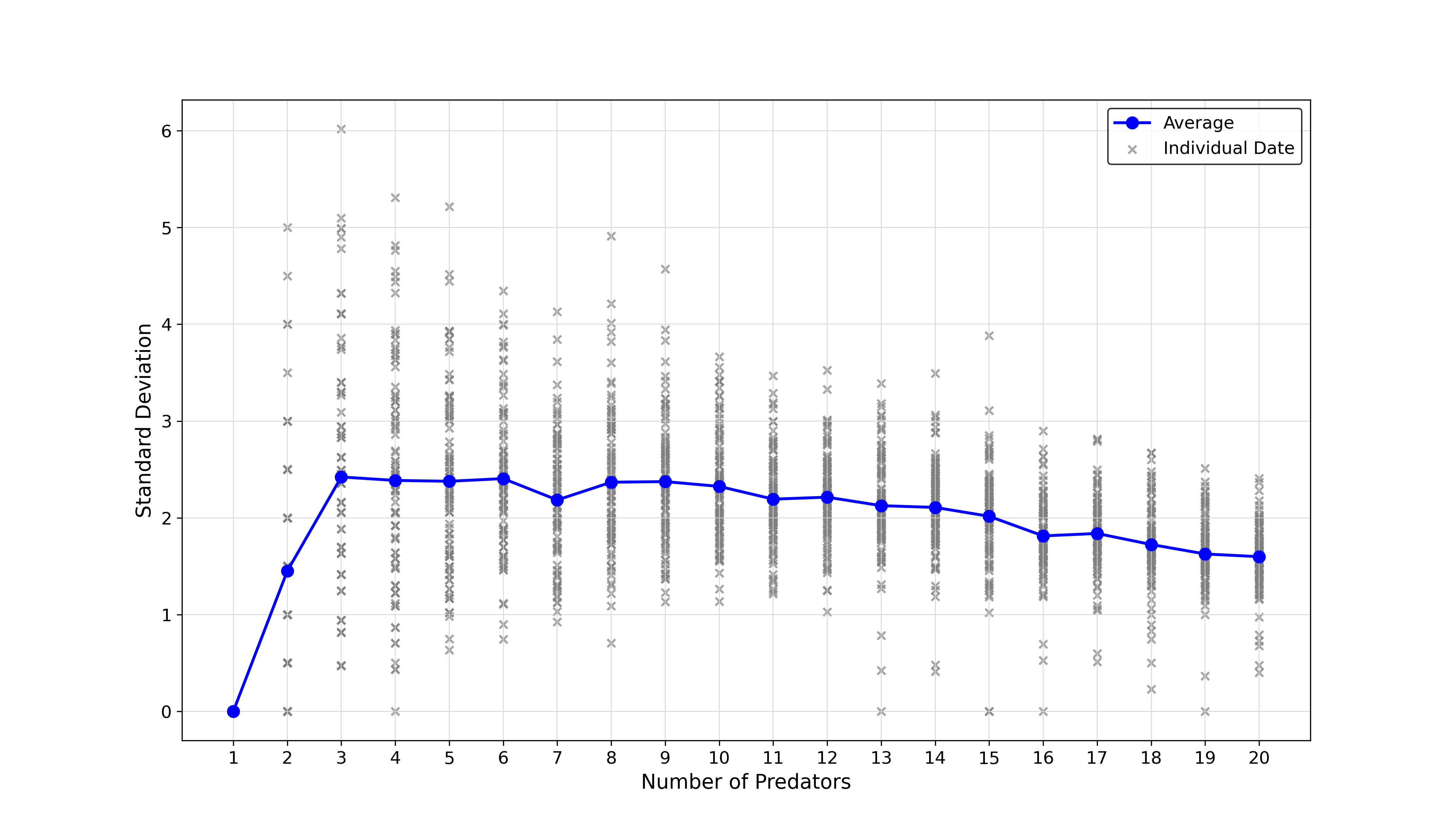}
    \caption{Standard Deviation of Average Prey Captured by Predators under Center Attack Strategy}
    \label{fig:center_std}
\end{figure}

Figure~\ref{fig:center_per} shows average prey captured per predator. Cooperative benefits are apparent for very small groups ($M \leq 3$), but efficiency declines monotonically thereafter. Extrapolating suggests that large predator schools would perform worse than solitary hunters.

\begin{figure}[H]
    \centering
    \includegraphics[scale=0.3]{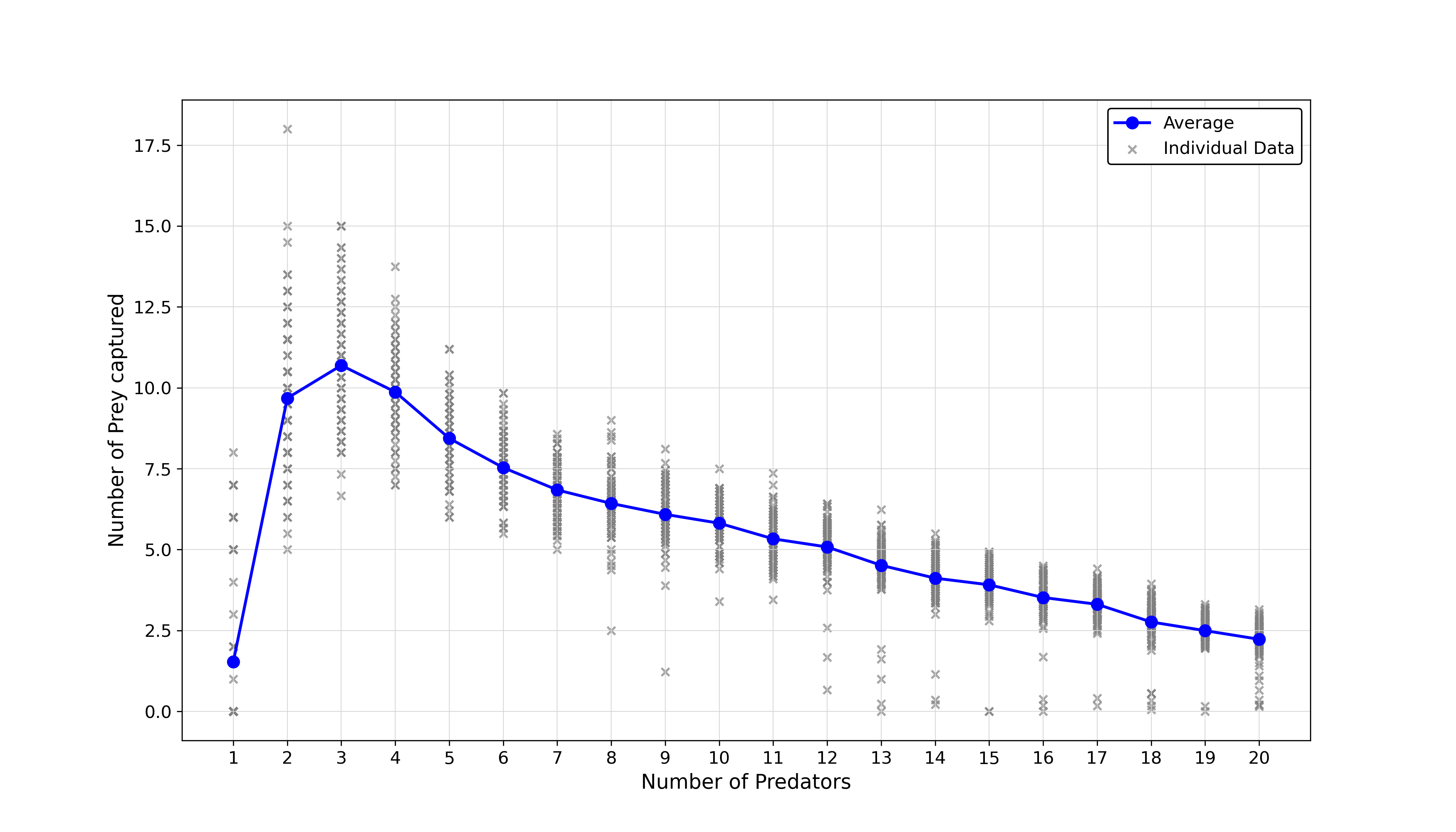}
    \caption{Average Number of Prey Captured per Predator under Center Attack Strategy}
    \label{fig:center_per}
\end{figure}

Prey survival time (Figure~\ref{fig:center_time}) decreases with predator numbers at first, but later stabilizes and even rebounds slightly around $M=17$. This rebound coincides with prey dispersal and increased predator interference. Survival rate (Figure~\ref{fig:center_rate}) shows a similar pattern: steep decline initially, then leveling off, and in some cases rising again when predator schools are very large.

\begin{figure}[H]
    \centering
    \includegraphics[scale=0.3]{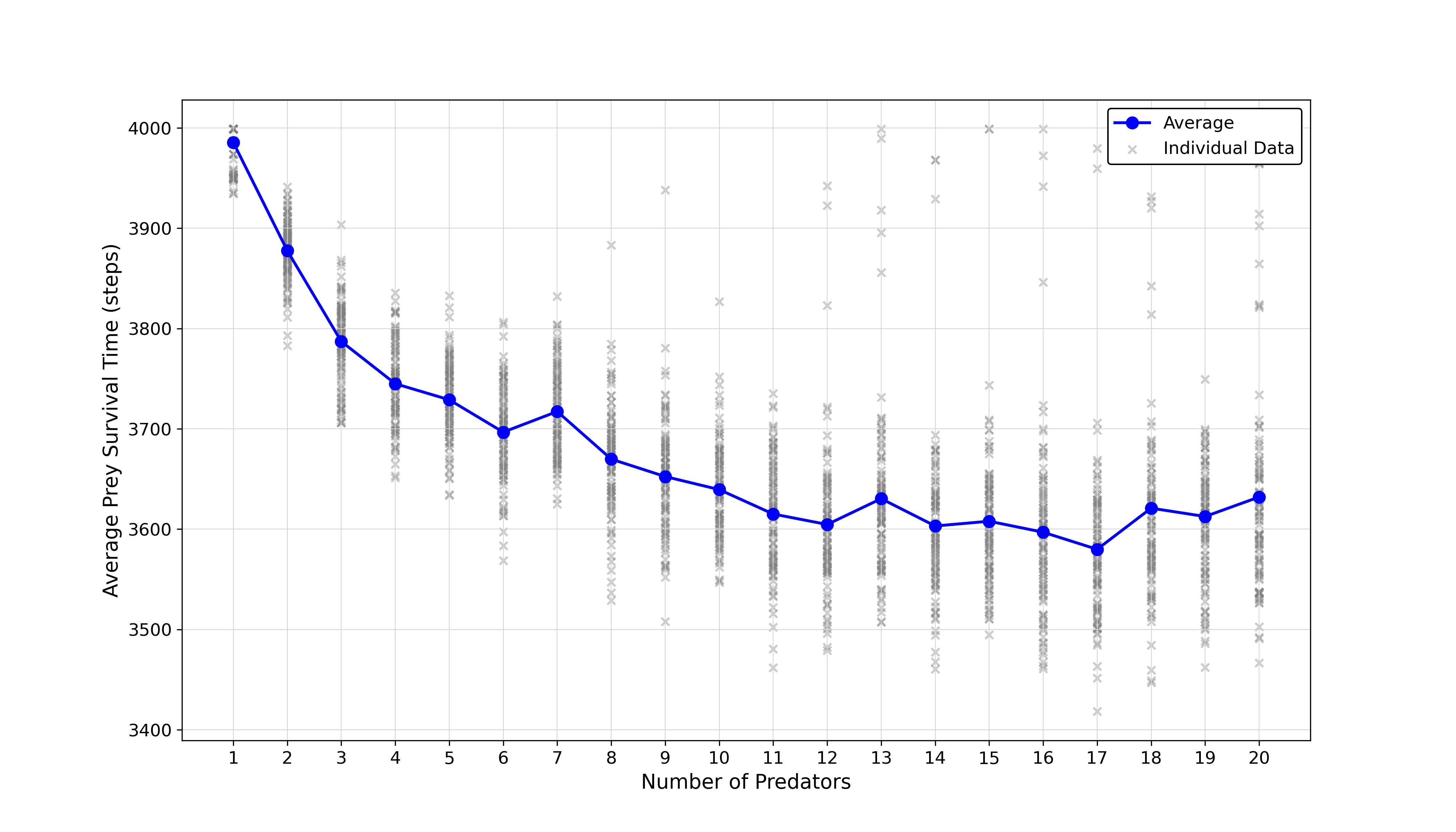}
    \caption{Prey Survival Time Across Different Predator Numbers under Center Attack Strategy}
    \label{fig:center_time}
\end{figure}

\begin{figure}[H]
    \centering
    \includegraphics[scale=0.3]{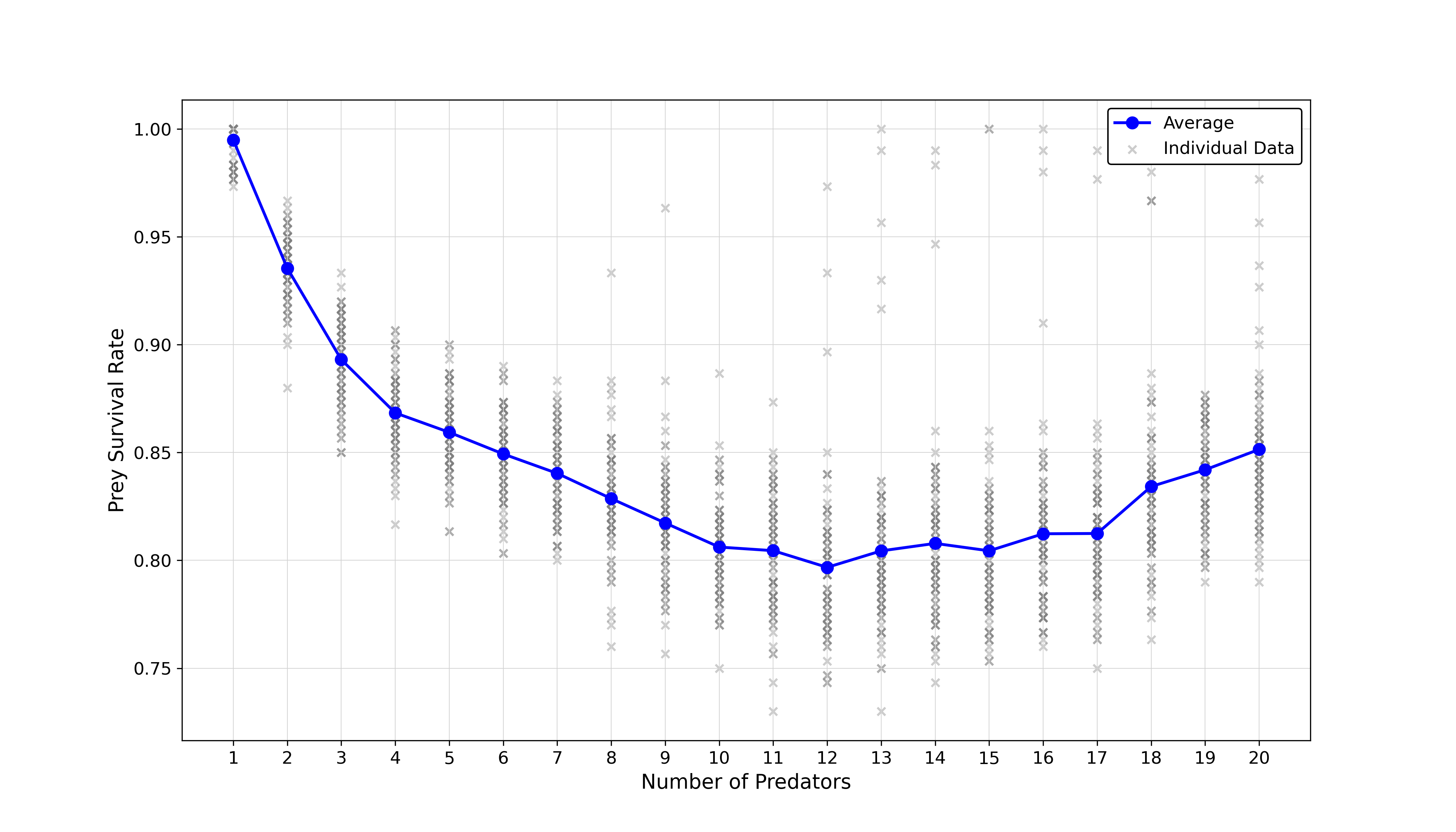}
    \caption{Prey Survival Rate Across Different Predator Numbers under Center Attack Strategy}
    \label{fig:center_rate}
\end{figure}

In summary, both strategies reveal strong nonlinear effects of predator school size. Small groups benefit from cooperation, but excessive numbers create competition and interference, limiting efficiency and in some cases improving prey survival.

\section{Conclusion}
\label{Conclusion}

This study develops a stochastic differential equation (SDE)-based model to explore predator–prey interactions in aquatic environments. By incorporating attraction, repulsion, alignment, and stochastic noise, the model reproduces realistic schooling behaviors and captures the complex dynamics between predator and prey groups. We examine how predator school size and hunting strategies influence prey survival, predation efficiency, and group-level behaviors.

Our results show that collective hunting increases predation efficiency compared to solitary attacks, but the benefit is nonlinear. Efficiency rises with predator number at first, then declines beyond a critical threshold due to competition and interference among predators. This highlights the dual role of cooperation and competition in shaping predator group success. Moreover, the model reproduces a variety of emergent dynamics—including prey dispersal and regrouping, cohesive defensive formations, and predator encirclement—that reveal the adaptive strategies both sides employ under predation pressure.

These findings advance understanding of how group size and hunting strategy shape predator–prey outcomes, with broader implications for ecological stability and the evolution of collective behaviors. In applied contexts, they indicate that managing predator–prey balances in aquatic ecosystems requires attention not only to species abundance but also to social hunting dynamics.

Future work could explore how these patterns shift under different model parameters and extend the framework to multiple predator species, higher trophic interactions, and heterogeneous environments, thereby providing deeper insights into the resilience and adaptability of natural ecosystems.

\section*{Acknowledgment}
This research was supported by the Tokuo Fujii Research Fund -- Support for Article Processing Charge of International Academic Papers.

\end{document}